\documentclass{llncs}
\usepackage[utf8]{inputenc}
\def\techreport{a}

\newcommand{\pagewidth}{\textwidth}

\usepackage{amsmath,amssymb,xcolor,ifthen,tikz,flushend,paralist,dsfont}
\usepackage[format=hang]{subfig}
\usepackage{fixltx2e}
\usepackage{etoolbox}
\usepgflibrary{shapes}
\usetikzlibrary{arrows,positioning,automata}
\tikzset{->,>=stealth'}
\allowdisplaybreaks

\newcommand{\thmhelperpre}[2]{\newcommand{\theoremlike}[1]{\par\medskip\penalty-250\refstepcounter{theorem}{\bfseries\noindent##1 \ref{#1}.}\itshape}\theoremlike{#2}}
\newcommand{\thmhelperpost}{\par\medskip%
 \renewcommand{\theoremlike}[1]{\par\medskip\penalty-250\refstepcounter{theorem}{\bfseries\noindent##1 \thesection .\thetheorem.}\itshape}%
}

\newenvironment{reflemma}[1]{\thmhelperpre{#1}{Lemma}}{\thmhelperpost}
\newenvironment{refproposition}[1]{\thmhelperpre{#1}{Proposition}}{\thmhelperpost}

\newcounter{applemma}
\newtheorem{alemma}[applemma]{Lemma}

\newcommand{\Qset}{\mathbb{Q}}
\newcommand{\Nset}{\mathbb{N}}

\newcommand{\dist}{\mathit{Dist}}

\newcommand{\X}{{\ensuremath{\mathbf{X}}}}
\newcommand{\F}{{\ensuremath{\mathbf{F}}}}
\newcommand{\G}{{\ensuremath{\mathbf{G}}}}
\newcommand{\Gf}[1]{\G^{#1}}
\newcommand{\U}{{\ensuremath{\mathbf{U}}}}
\newcommand{\true}{{\ensuremath{\mathbf{tt}}}}
\newcommand{\false}{{\ensuremath{\mathbf{ff}}}}
\newcommand{\fltlfragment}{{\textrm{fLTL}\textsubscript{$\setminus\G\U$}}}

\renewcommand{\sf}{{\ensuremath{\mathsf{sf}}}}

\newcommand{\un}{{\ensuremath{\mathsf{Unf}}}}

\newcommand{\gsf}{{\ensuremath{\mathbb{G}}}}
\newcommand{\fsf}{{\ensuremath{\mathbb{F}}}}

\newcommand{\tran}[1]{\stackrel{#1}{\longrightarrow}}
\newcommand{\acc}{{\ensuremath{\mathit{Acc}}}}
\newcommand{\sink}{\mathsf{Sink}}

\newcommand{\SGF}{\mathcal{S}_{\G\F}}
\newcommand{\SFG}{\mathcal{S}_{\F\G}}
\newcommand{\SGf}{\mathcal{S}_{\Gf{\bowtie p}_\ext}}
\newcommand{\sat}{\mathit{Sat}}

\newcommand{\accInfinite}{\Inf}
\newcommand{\accFinite}{\Fin}
\DeclareMathOperator{\limext}{\mathrm{lim\,ext}}
\DeclareMathOperator{\Inf}{\mathit{Inf}}
\DeclareMathOperator{\Fin}{\mathit{Fin}}
\newcommand{\accMPname}{\mathit{MP}}
\newcommand{\accMP}[3][]{\ifblank{#1}{\accMPname^{#2}(#3)}{\accMPname_{#1}^{#2}(#3)}} %
\newcommand{\accmec}{\textsc{AcceptingMEC}}

\newcommand{\accT}[1]{\acc_{\T}({#1})}
\newcommand{\accP}[1]{\acc({#1})}
\newcommand{\T}{\mathcal{M}}
\renewcommand{\S}{\mathcal{S}}

\newcommand{\A}{\mathcal{A}}
\renewcommand{\P}{\mathcal{P}}
\newcommand{\rec}{\mathcal{R}\mathrm{ec}}
\newcommand{\R}{\mathcal{R}}
\renewcommand{\L}{\mathsf{L}}

\newcommand{\suffix}[2]{#1^{#2}}%
\newcommand{\infix}[3]{#1^{#2#3}}

\newcommand{\mdp}{\mathsf{M}}
\newcommand{\product}{\mdp\times\A}
\newcommand{\inits}{\hat s}
\newcommand{\act}[1]{\mathit{Act}(#1)}
\newcommand{\pat}{\omega}

\newcommand{\fpat}{w}

\newcommand{\mec}{\mathsf{MEC}}

\newcommand{\pr}{\mathbb P}	
\renewcommand{\Pr}[3]{\pr^{#1}_{#2}\hspace{-0.16em}\left[{#3}\right]}   %

\newcommand{\singlereward}{r}

\newcommand{\singlerewardalt}{q}

\newcommand{\lrLim}[1]{\mathrm{lr}(#1)}  %
\newcommand{\lrInf}{\mathrm{lr}_{\inf}}  %
\newcommand{\lrSup}{\mathrm{lr}_{\sup}}
\newcommand{\lrExt}{\mathrm{lr}_{\ext}}
\newcommand{\ext}{\mathrm{ext}}

\DeclareMathOperator{\proves}{\vdash}
\DeclareMathOperator{\notproves}{\not\vdash}

\newcommand{\idf}{\mathds{1}}
\newcommand{\mypara}[1]{\smallskip\noindent {\bf #1}\ }

\allowdisplaybreaks

\usepackage{microtype}
\newcommand{\myspace}{\vspace*{-0.5em}}

\begin{document}
\title{Controller synthesis for MDPs and \\ Frequency LTL$_{\setminus \G\U}$}

\author{Vojt\v{e}ch~Forejt\inst{1} \and Jan Kr\v{c}\'al\inst{2} \and Jan K\v{r}et\'insk\'y\inst{3}}
\institute{Department of Computer Science, University of Oxford, UK\and Saarland University -- Computer Science,  Saarbr\"{u}cken, Germany \and IST Austria}

\maketitle

\begin{abstract}
Quantitative extensions of temporal logics have recently attracted significant attention. In this work, we study frequency LTL (fLTL), an extension of LTL which allows to speak about frequencies of events along an execution. 
Such an extension is particularly useful for probabilistic systems that often cannot fulfil strict qualitative guarantees on the behaviour.
It has been recently shown that controller synthesis for Markov decision processes and fLTL is decidable when all the bounds on frequencies are $1$. As a step towards a complete quantitative solution, we show that the problem is decidable for the fragment \fltlfragment, where $\U$ does not occur in the scope of $\G$ (but still $\F$ can).
Our solution is based on a novel translation of such quantitative formulae into equivalent deterministic automata.
\end{abstract}

\section{Introduction}

Markov decision processes (MDP) are a common choice when modelling systems that exhibit (un)controllable and probabilistic behaviour.
In controller synthesis of MDPs, the goal is then to steer the system so that it meets certain property. 
Many properties specifying the desired behaviour, such as ``the system is always responsive'' can be easily captured by
Linear Temporal Logic (LTL).
This logic is in its nature qualitative and cannot express \emph{quantitative} linear-time properties such as ``a given failure happens only {\em rarely}''.
To overcome this limitation, especially apparent for stochastic systems, extensions of LTL with \emph{frequency} operators have been recently studied~\cite{BDL-tase12,BMM14}.
Such extensions come at a cost, and for example the ``frequency until'' operator can make the controller-synthesis problem undecidable already for non-stochastic systems~\cite{BDL-tase12,BMM14}. 
It turns out~\cite{our-concur,DBLP:journals/corr/abs-1111-3111,AT12} that a way of providing significant added expressive power while preserving tractability is to extend LTL only by the ``frequency globally'' formulae $\Gf{\geq p} \varphi$.
Such a formula is satisfied if the long-run frequency of satisfying $\varphi$ on an infinite path is at least $p$. 
More formally, $\Gf{\geq p} \varphi$ is true on an infinite path $s_0s_1 \cdots$ of an MDP if and only if $\frac{1}{n}\cdot |\{i \mid \text{$i < n$ and $s_i s_{i+1} \cdots$ satisfies $\varphi$}\}|$ is at least $p$ as $n$ tends to infinity. Because the relevant limit might not be defined, we need to consider two distinct operators, $\Gf{\geq p}_{\inf}$ and $\Gf{\geq p}_{\sup}$, whose definitions
use limit inferior and limit superior, respectively.
We call the resulting logic \emph{frequency LTL (fLTL)}. 

So far, MDP controller synthesis for fLTL has been shown decidable for the fragment containing only the operator $\Gf{\geq 1}_{\inf}$~\cite{our-concur}.
Our paper makes a significant further step towards the ultimate goal of a model checking procedure for the whole fLTL. We address the general \emph{quantitative} setting with arbitrary frequency bounds $p$ and consider %
the fragment \fltlfragment{}, which is obtained from frequency LTL by preventing the $\U$ operator from occurring inside $\G$ or $\Gf{\geq p}$ formulas (but still allowing the $\F$ operator to occur anywhere in the formula).
The approach we take is completely different from~\cite{our-concur} 
where %
ad hoc product MDP construction is used, heavily relying on existence of certain types of strategies in the $\Gf{\geq 1}_{\inf}$ case.
In this paper we provide, to the best of our knowledge, the first translation of a quantitative logic to equivalent \emph{deterministic} automata. 
This allows us to take the standard automata-theoretic approach to verification~\cite{VW86}: after obtaining the finite automaton, we do not deal with the structure of the formula originally given, and we solve a (reasonably simple) synthesis problem on a product of the single automaton with the MDP.

Relations of various kinds of logics and automata are widely studied (see e.g.~\cite{Vardi96anautomata-theoretic,thomas1997languages,Droste}), and our results provide new insights into this area 
for quantitative logics.
Previous work~\cite{AT12} offered only translation of a similar logic to {\em non-deterministic} ``mean-payoff B\"uchi automata'' noting that it is difficult to give an analogous
reduction to {\em deterministic} ``mean-payoff Rabin automata''. The reason is that the non-determinism is inherently present in the form of guessing whether the subformulas of $\Gf{\geq p}$ are satisfied on a suffix.
Our construction overcomes this difficulty and offers equivalent deterministic automata. It is a first and highly non-trivial step towards providing a reduction for the complete logic. \
Although our algorithm does not allow us to handle the extension of the whole LTL, the considered fragment \fltlfragment{} 
contains a large class of formulas and offers significant expressive power.
It subsumes the GR(1) fragment of LTL  \cite{BJPPS12}, which has found use in synthesis for hardware designs. The $\U$ operator, although not allowed within a scope of a $\G$ operator, can still
be used for example to distinguish paths based on their prefixes. As an example synthesis problem expressible in this fragment, 
consider a cluster of servers 
where each server
plays either a role of a load-balancer or a worker. 
On startup, each server 
\underline{l}istens
for a message specifying its role. A load-%
\underline{b}alancer
\underline{f}orwards
each 
\underline{r}equest
and only waits for a 
\underline{c}onfirmation
whereas a 
\underline{w}orker
\underline{p}rocesses
the requests itself.
A specification for a single server in the cluster can require, for example, that the following formula (with propositions \underline{explained} above) holds with probability at least $0.95$:
\[
\Big(\big(\mathit{l}\,\U\,\mathit{b}\big) \rightarrow \Gf{\geq 0.99}\big(r \rightarrow \X (f\wedge \F c)\big)\Big)
\wedge
\Big(\big(\mathit{l}\,\U\,\mathit{w}\big) \rightarrow \Gf{\geq 0.85}\big(r \rightarrow (\X p \vee \X\X p)\big)\Big)
\]

\mypara{Related work.} 
Frequency LTL was studied in another variant in~\cite{BDL-tase12,BMM14} where a {\em frequency until} operator is introduced in two different LTL-like logics, and undecidability is proved for problems relevant to our setting. The work \cite{BDL-tase12} also yields decidability with restricted nesting of the frequency until operator; as the decidable fragment in~\cite{BDL-tase12} does not contain frequency-globally operator, it is not possible to express many useful properties expressible in our logic.
A logic that speaks about frequencies on a finite interval was introduced in~\cite{DBLP:journals/corr/abs-1111-3111}, but the paper provides algorithms only for Markov chains and a bounded fragment of the logic.

Model checking MDPs against LTL objectives relies on the automata-theoretic approach, namely on translating LTL to automata that are to some extent deterministic~\cite{CY95}.
This typically involves translating LTL to non-deterministic automata, which are then determinized using e.g. Safra's construction.
During the determinization, the original structure of the formula is lost, which prevents us from extending this technique to the frequency setting.
However, an alternative technique of translating LTL directly to deterministic automata has been developed \cite{cav12,atva13,cav14}, where the logical structure is preserved.
In our work, we extend the algorithm for LTL$_{\setminus\G\U}$ partially sketched in \cite{atva13}.
In %
Section~\ref{sec:conclusions}, 
we explain why adapting the algorithm for full LTL \cite{cav14} is difficult.
Translation of LTL$_{\setminus\G\U}$ to other kinds of automata has been considered also in \cite{DBLP:conf/tacas/KiniV15}.
Our technique relies on a solution of a multi-objective mean-payoff problem on MDP~\cite{BBC+14,lics15}. Previous results only consider limit inferior rewards, and so we cannot use them as off-the-shelf results, but need to adapt them first to our setting with both inferior and superior limits together with Rabin condition.
There are several works that combine mean-payoff objectives with e.g. logics or parity objectives, but in most cases only simple atomic propositions can be used to define the payoff \cite{bloem2009better,boker2011temporal,chatterjee2011energy}.
The work \cite{baier2014weight} extends LTL with another form of quantitative operators, allowing accumulated weight constraint expressed using automata, again not allowing quantification over complex formulas. Further, \cite{ABK14} introduces a variant of LTL with a discounted-future operator.
Finally, techniques closely related to the ones in this paper are used in \cite{EKVY08,CR15lics,RRS15cav}.

\mypara{Our contributions.}
To our best knowledge, this paper gives the first decidability result for probabilistic verification against linear-time temporal logics extended by \emph{quantitative} frequency operators with \emph{complex nested subformulas} of the logic.
It works in two steps,
keeping the same time complexity as for ordinary LTL.
In the first step, a $\fltlfragment$ formula gets translated to an equivalent \emph{deterministic} generalized Rabin automaton extended with mean-payoff objectives. 
 This step is inspired by previous work~\cite{atva13}, %
 but the extension with auxiliary automata for $\Gf{\geq p}$ requires a different construction.
The second step is the analysis of MDPs against conjunction of limit inferior mean-payoff, limit superior mean-payoff, and generalized Rabin objectives.
 This result is obtained by adapting and combining several existing involved proof techniques~\cite{lics15,BCFK13}.

The paper is organised as follows: 
the main algorithm is explained in Section~\ref{sec:alg}, relegating the details of the two technical steps above to Sections~\ref{sec:automata} and~\ref{sec:mean-payoff}.

\section{Preliminaries}

We use  $\Nset$ and $\Qset$ to denote the sets of non-negative integers and rational numbers.
The set of all distributions over a countable set $X$ is denoted by $\dist(X)$.
For a predicate $P$, the {\em indicator function} $\idf_{P}$ equals $1$ if $P$ is true, and $0$ if $P$ is false.

\mypara{Markov decision processes (MDPs).}
An MDP is a tuple $\mdp=(S,A,\mathit{Act},\delta,\inits)$ 
where $S$ is a finite set of states, $A$ is a finite set of actions, 
$\mathit{Act} : S\rightarrow 2^A\setminus \{\emptyset\}$ assigns to each state $s$ the set $\act{s}$ of actions enabled 
in $s$,
$\delta : A\rightarrow \dist(S)$ is a probabilistic 
transition function that given 
an action $a$ 
gives a probability distribution over the 
successor states, and $\inits$ is the initial state.
To simplify notation, w.l.o.g. we require that every action is enabled in exactly one state.

\mypara{Strategies.} 
A strategy in an MDP $\mdp$ is a ``recipe'' to choose actions.
Formally, it is a function 
$\sigma : (SA)^*S \to \dist(A)$ that given a finite path~$\fpat$, representing 
the history of a play, gives a probability distribution over the 
actions enabled in the last state. 
A strategy $\sigma$ in $\mdp$ induces a \emph{Markov chain} 
$\mdp^\sigma$ which is a tuple $(L,P,\inits)$ where the set of \emph{locations} $L=(S \times A)^*\times S$ encodes the history of the play,
$\inits$ is an \emph{initial location},
and $P$ is a \emph{probabilistic transition function} that assigns to each location a probability distribution over successor locations defined by
\(
   P(h)(h\,a\,s)\ =\ 
   \sigma(h)(a)\cdot \delta(a)(s)\,.
\)
for all $h\in (SA)^*S$, $a\in A$ and $s\in S$.

The probability space of the runs of the Markov chain is denoted by $\pr^\sigma_\mdp$ and defined in the standard way
\ifx\techreport\undefined\cite{KSK76,techreport}\else \cite{KSK76}; for reader's convenience the construction is recalled in Appendix~\ref{app:mc}\fi.

\mypara{End components.}
A tuple $(T,B)$ with $\emptyset\neq T\subseteq S$ and $B\subseteq \bigcup_{t\in T}\act{t}$
is an \emph{end component} of $\mdp$
if (1) for all $a\in B$, whenever $\delta(a)(s')>0$ then $s'\in T$;
and (2) for all $s,t\in T$ there is a path 
$w = s_1 a_1\cdots a_{k-1} s_k$ such that $s_1 = s$, $s_k=t$, and all states
and actions that appear in $w$ belong to $T$ and $B$, respectively.
An end component $(T,B)$ is a \emph{maximal end component (MEC)}
if it is maximal with respect to the componentwise subset ordering. Given an MDP, the set of MECs is denoted by $\mec$.
Finally, an MDP is \emph{strongly connected} if $(S,A)$ is a MEC.

\mypara{Frequency linear temporal logic (fLTL).}
The formulae of the logic
fLTL
are given by the following syntax:
\[
\varphi\quad\mathop{::=}\quad \true \mid \false \mid a\mid \neg a\mid \varphi\wedge\varphi \mid \varphi\vee\varphi \mid \X\varphi \mid \F\varphi \mid \G\varphi \mid \varphi\U\varphi  \mid \Gf{\bowtie p}_{\ext} \varphi 
\] 
over a finite set $Ap$ of atomic propositions, ${\bowtie}\in\{\geq,>\}$, $p\in[0,1]\cap\Qset$, and $\ext \in \{\inf,\sup \}$.
A formula that is neither a conjunction, nor a disjunction is called \emph{non-Boolean}.
The set of non-Boolean subformulas of $\varphi$ is denoted by $\sf(\varphi)$.

\mypara{Words and fLTL Semantics.}
Let $w\in (2^{Ap})^\omega$ be an infinite word. The $i$th letter of $w$ is denoted $w[i]$, i.e.~$w=w[0]w[1]\cdots$. 
We write $\infix w i j$ for the finite word $w[i] w[i+1] \cdots w[j]$, and $w^{i\infty}$
or just $\suffix w i$ for the suffix $w[i] w[i+1] \cdots $. 
The semantics of a formula on a word $w$ is defined inductively: for $\true$, $\false$, $\wedge$, $\vee$, and for
atomic propositions and their negations, the definition is straightforward, for the remaining operators
we define:
$$
\begin{array}[t]{lcl}
w \models \X \varphi & \iff & \suffix w1 \models \varphi \\
w \models \F \varphi & \iff & \exists \, k\in\Nset: \suffix wk \models \varphi \\
w \models \G \varphi & \iff & \forall \, k\in\Nset: \suffix wk \models \varphi
\end{array}
\quad
\begin{array}[t]{lcl}
w \models \varphi\U \psi & \iff &
\begin{array}[t]{l}
\exists \, k\in\Nset: \suffix wk \models \psi \text{ and } \\
\forall\, 0\leq j < k: \suffix wj\models \varphi
\end{array}\\
w \models \Gf{\bowtie p}_{\ext} \varphi & \iff & \lrExt(\idf_{\suffix w0\models\varphi}\idf_{\suffix w1\models\varphi}\cdots) \bowtie p 
\end{array}$$
where we set $\lrExt(q_1q_2\cdots) := \limext_{i\to \infty} \frac{1}{i} \sum_{j=1}^i q_i$.
By $\L(\varphi)$ we denote the set $\{w\in(2^{Ap})^\omega\mid w\models\varphi\}$ of words satisfying $\varphi$.

The \fltlfragment{} fragment of fLTL is defined by disallowing occurrences of $\U$ in $\G$-formulae, i.e. it is given by the following syntax for $\varphi$:
\begin{align*}
\varphi::= & a\mid \neg a\mid \varphi\wedge\varphi \mid \varphi\vee\varphi \mid \X\varphi \mid \varphi\U\varphi \mid \F\varphi \mid \G\xi \mid \Gf{\bowtie p}_{\ext} \xi\\
\xi::= & a\mid \neg a\mid \xi\wedge\xi \mid \xi\vee\xi \mid \X\xi \mid  \F\xi \mid \G\xi \mid \Gf{\bowtie p}_{\ext} \xi
\end{align*} 

Note that restricting negations to atomic propositions is without loss of generality as all operators are closed under negation, for example $\neg \Gf{\ge p}_{\inf} \varphi \equiv \Gf{> 1- p}_{\sup} \neg \varphi$ or $\neg \Gf{> p}_{\sup} \varphi \equiv \Gf{\ge 1- p}_{\inf} \neg \varphi$.  Furthermore, we could easily allow $\bowtie$ to range also over $\leq$ and $<$ as $\Gf{\leq p}_{\inf} \varphi \equiv \Gf{\geq 1- p}_{\sup} \neg \varphi$ and $\Gf{< p}_{\inf} \varphi \equiv \Gf{> 1- p}_{\sup} \neg \varphi$. 

\mypara{Automata.}
Let us fix a finite alphabet $\Sigma$.
A deterministic \emph{labelled transition system (LTS)} over $\Sigma$ is a tuple $(Q,q_0,\delta)$ where $Q$ is a finite set of states, $q_0$ is the initial state, and $\delta: Q \times \Sigma \to Q$ is a partial transition function. We denote $\delta(q,a) = q'$ also by $q \tran{a} q'$.
A \emph{run} of the LTS $\S$ over an infinite word $w$ is a sequence of states $\S(w)=q_0 q_1 \cdots$ such that $q_{i+1} = \delta(q_i,w[i])$.
For a finite word $w$ of length $n$, we denote by $\S(w)$ the state $q_n$ in which $\S$ is after reading $w$.

An \emph{acceptance condition} is a positive boolean formula over 
formal 
variables 
\[
\{ \accInfinite(S), \accFinite(S), \accMP[\ext]{\bowtie p}{\singlereward} \mid S{\subseteq} Q, \ext {\in} \{\inf,\sup\}, \mathord{\bowtie} {\in} \{ {\geq}, {>}\}, p {\in} \Qset, \singlereward: Q {\to} \Qset \}.
\]
Given a run $\rho$ and an acceptance condition $\alpha$, we 
assign truth values as follows:
\begin{itemize}
	\item $\accInfinite(S)$ is true if{}f $\rho$ visits 
	(some state of) $S$ infinitely often,
	\item $\accFinite(S)$ is true
	if{}f $\rho$ visits (all states of) $S$ finitely often,
	\item $\accMP[\ext]{\bowtie p}{\singlereward}$ is true if{}f 
	$\lrExt(\singlereward(\rho[0]) \singlereward(\rho[1]) \cdots) \bowtie p$.
\end{itemize}
The run $\rho$ satisfies $\alpha$ 
if this truth-assignment makes $\alpha$ true.
An \emph{automaton} $\A$ is an LTS with an acceptance condition $\alpha$. The language of $\A$, denoted by $\L(\A)$, is the set of all words inducing a run satisfying $\alpha$.
An acceptance condition $\alpha$ is a \emph{B\"uchi}, \emph{generalized B\"uchi}, or \emph{co-B\"uchi} acceptance condition
if it is of the form $\accInfinite(S)$, $\bigwedge_i\accInfinite(S_i)$, or $\accFinite(S)$, respectively.
Further, $\alpha$ is a \emph{generalized Rabin mean-payoff}, or a \emph{generalized B\"uchi mean-payoff} acceptance condition
if it is in disjunctive normal form, or if it is a conjunction not containing any $\accFinite(S)$, respectively.
For each acceptance condition we define a corresponding automaton, e.g. \emph{deterministic generalized Rabin mean-payoff automaton (DGRMA)}.

\section{Model-checking algorithm}\label{sec:alg}

In this section, we state the problem of model checking MDPs against $\fltlfragment$ specifications and provide a solution.
As a black-box we use two novel routines described in detail in the following two sections. All proofs are in the appendix.

Given an MDP $\mdp$ and a valuation $\nu:S\to2^{Ap}$ of its states, we say that its run 
$\pat= s_0 a_0 s_1 a_1 \cdots$ \emph{satisfies} $\varphi$, written $\pat\models\varphi$, if $\nu(s_0)\nu(s_1)\cdots\models \varphi$.
We use $\Pr{\sigma}{}{\varphi}$ as a shorthand for the probability of all runs satisfying $\varphi$, i.e.  $\Pr{\sigma}{\mdp}{\{\omega\mid \omega\models \varphi\}}$.
This paper is concerned with the following task:\medskip

\noindent
\framebox[\pagewidth]{\parbox{0.96\pagewidth}{\smallskip
		
		\textbf{Controller synthesis problem:}
		Given an MDP with a valuation, an $\fltlfragment$ formula $\varphi$ and
		$x\in[0,1]$, decide whether $\Pr{\sigma}{}{\varphi} \ge x$ for some strategy $\sigma$,
		and if so, construct such a \emph{witness} strategy.
		\smallskip
	}}\medskip

\noindent
The following is the main result of the paper.
\begin{theorem}\label{thm:main}
The controller synthesis problem for MDPs and \fltlfragment{} is decidable and the witness strategy can be constructed in doubly exponential time.
\end{theorem}
In this section, we present an algorithm for Theorem~\ref{thm:main}. %
The skeleton of our algorithm is the same as for the standard model-checking algorithm for MDPs against LTL. 
It proceeds in three steps. Given an MDP $\mdp$ and a 
formula $\varphi$,
\begin{enumerate}
 \item compute a deterministic automaton $\A$ such that $\L(\A)=\L(\varphi)$,
 \item compute the product MDP $\product$,
 \item analyse the product MDP $\product$.
\end{enumerate}
In the following, we concretize these three steps to fit our setting. 

\mypara{1.~Deterministic automaton}
For ordinary LTL, usually a Rabin automaton or a generalized Rabin automaton is constructed~\cite{prism,ltl2dstar,cav14,atva14}.
Since in our setting, along with $\omega$-regular language the specification also includes quantitative constraints over runs, we generate a DGRMA. The next theorem is the first black box, detailed in Section~\ref{sec:automata}.

\begin{theorem}\label{thm:auto}
For any $\fltlfragment$ formula, there is a DGRMA $\A$, constructible in doubly exponential time, such that $\L(\A)=\L(\varphi)$, and the acceptance condition is of exponential size.
\end{theorem}

\mypara{2.~Product}
Computing the synchronous parallel product of the MDP $\mdp=(S,A,\mathit{Act},\Delta,\inits)$ with valuation $\nu:S\to 2^{Ap}$ and the LTS $(Q,i,\delta)$ over $2^{Ap}$ underlying $\A$ is rather straightforward.
The product $\product$ is again an MDP $(S\times Q, A\times Q,\mathit{Act}',\Delta',(\inits,\hat{q}))$ where\footnote{In order to guarantee that each action is enabled in at most one state, we have a copy of each original action for each state of the automaton.} $\mathit{Act}'((s,q))=\mathit{Act}(s)\times\{q\}$, $\hat{q} = \delta(i,\nu(\inits))$, and 
$\Delta'\big((a,q)\big)\big(({s},\bar{q})\big)$ is equal to $\Delta(a)(s)$ if $\delta(q,\nu({s}))=\bar{q}$, and to $0$ otherwise.
We lift acceptance conditions $\acc$ of $\A$ to $\product$: a run of $\product$ satisfies $\acc$ if its projection to the component of the automata states satisfies $\acc$.\footnote{Technically, the projection should be preceded by $i$ to get a run of the automaton, but the acceptance does not depend on any finite prefix of the sequence of states.} 

\mypara{3.~Product analysis} The MDP $\product$ is solved with respect to $\acc$, i.e., a strategy in $\product$ is found that maximizes the probability of satisfying $\acc$.
Such a strategy then induces a (history-dependent) strategy on $\mdp$ in a straightforward manner.
Observe that for DGRMA, it is sufficient to consider the setting with 
\begin{align}\label{eq:alg-acc}
\acc=\bigvee_{i=1}^k(\Fin(F_i)\wedge\acc_i')
\end{align}
where $\acc_i'$ is a conjunction of several $\Inf$ and $\accMPname$ (in contrast with a Rabin condition used for ordinary LTL where $\acc_i'$ is simply of the form $\Inf(I_i)$). 
Indeed, one can replace each $\bigwedge_{j}\Fin(F_j)$ by $\Fin(\bigcup_j F_j)$ to obtain the desired form, since avoiding several sets is equivalent to avoiding their union.

For a condition of the form (\ref{eq:alg-acc}), the solution is obtained as follows:
\begin{enumerate}
\item For $i=1,2,\ldots, k$:
\begin{enumerate}
\item Remove the set of states $F_i$ from the MDP.
\item Compute the MEC decomposition.
\item Mark each MEC $C$ as winning iff $\accmec(C,\acc_i')$ returns Yes. 
\item Let $W_i$ be the componentwise union of winning MECs above.
\end{enumerate}
\item Let $W$ be the componentwise union of all $W_i$ for $1\le i\le k$.
\item Return the 
maximal probability to reach the set $W$ in the MDP.
\end{enumerate}
The procedure $\accmec(C,\acc_i')$ is the second black box used in our algorithm, detailed in Section~\ref{sec:mean-payoff}. It decides, whether the maximum probability of satisfying $\acc_i'$ in $C$ is $1$ (return Yes), or $0$ (return No).

\begin{theorem}\label{thm:mec}
For a strongly connected MDP $\mdp$ and a generalized B\"uchi mean-payoff acceptance condition $\acc$, %
the maximal probability to satisfy $\acc$ is either $1$ or $0$, and is the same for all initial states.
Moreover, there is a polynomial-time algorithm that computes this probability, and also outputs a witnessing strategy if the probability is $1$.
\end{theorem}
The procedure is rather complex in our case, as opposed to 
standard cases such as Rabin condition, where a MEC is accepting for $\acc_i'=\Inf(I_i)$ if its states intersect $I_i$; or
a generalized Rabin condition~\cite{cav13}, where a MEC is accepting for $\acc_i'=\bigwedge_{j=1}^{\ell_i}\Inf(I_{ij})$ if its states
intersect with each $I_i^j$, for $j=1,2,\ldots,\ell_i$.
\smallskip

\mypara{Finishing the proof of Theorem~\ref{thm:main}}
Note that for MDPs that are not strongly connected, the maximum probability might not be in $\{0,1\}$.
Therefore, the problem is decomposed into a qualitative satisfaction problem in step 1.(c) and a quantitative reachability problem in step 3. 
Consequently, the proof of correctness is the same as the proofs for LTL via Rabin automata \cite{BP08} and generalized Rabin automata \cite{cav13}.
The complexity follows from Theorem~\ref{thm:auto} and~\ref{thm:mec}. Finally, the overall witness strategy first reaches the winning MECs and if they are reached it switches to the witness strategies from Theorem~\ref{thm:mec}.

\begin{remark}
We remark that by a simple modification of the product construction above and of the proof of Theorem~\ref{thm:mec}, we obtain an algorithm synthesising a strategy achieving a given bound w.r.t.\ multiple mean-payoff objectives (with a combination of superior and inferior limits) and (generalized) Rabin acceptance condition for \emph{general} (not necessarily strongly connected) MDP.
\end{remark}

\tikzset{
			state/.style={
				rectangle,
				rounded corners,
				draw=black,
				minimum height=1.5em,
				minimum width=2em,
				inner sep=2pt,
				text centered,
			}
		}

\section{Automata characterization of \fltlfragment}
\label{sec:automata}

In this section, we prove Theorem~\ref{thm:auto}. We give an algorithm for translating a given $\fltlfragment$ formula $\varphi$ into a deterministic generalized Rabin mean-payoff automaton $\A$ that recognizes words satisfying $\varphi$.
For the rest of the section, let $\varphi$ be an $\fltlfragment$ formula. 
Further, $\fsf$, $\gsf$, $\gsf^{\bowtie}$, and $\sf$ denote the set of $\F$-, $\G$-, $\Gf{\bowtie p}_{\ext}$-, and non-Boolean subformulas of $\varphi$, respectively. 

In order to obtain an automaton for the formula, we first need to give a more operational view on fLTL. 
To this end, we use expansions of the formulae in a very similar way as they are used, for instance, in tableaux techniques for LTL translation to automata, or for deciding LTL satisfiability.
We define a symbolic one-step unfolding (expansion) $\un$ of a formula inductively by the rules below. 
Further, for a valuation $\nu\subseteq Ap$, we define the ``next step under $\nu$''-operator. This operator (1) substitutes unguarded atomic propositions for their truth values, and (2) peels off the outer $\X$-operator whenever it is present. Formally, we define

\noindent
\begin{minipage}{0.5\linewidth}
	\begin{align*}
		\un(\psi_1\wedge\psi_2)&=\un(\psi_1)\wedge\un(\psi_2)\\ 
		\un(\psi_1\vee\psi_2)&=\un(\psi_1)\vee\un(\psi_2)\\
		\un(\F\psi_1)&=\un(\psi_1)\vee\X\F\psi_1\\
		\un(\G\psi_1)&=\un(\psi_1)\wedge\X\G\psi_1\\
		\un(\psi_1\U\psi_2)&=\un(\psi_2){\vee}\big(\un(\psi_1){\wedge} \X(\psi_1\U\psi_2)\big)\\
		\un(\Gf{\bowtie p}_\ext\psi_1)&=\true\wedge\X\Gf{\bowtie p}_\ext\psi_1\\
		\un(\psi) &= \psi \text{ for any other $\psi$}\\
	\end{align*}
\end{minipage}
\hspace*{-1em}
\begin{minipage}{0.4\linewidth}
	\begin{align*}
		(\psi_1\wedge \psi_2)[\nu]&= \psi_1[\nu] \wedge \psi_2[\nu]\\
		(\psi_1\vee \psi_2)[\nu]&= \psi_1[\nu] \vee \psi_2[\nu]\\
		a[\nu]&=
		\begin{cases}
			\true &\text{if }a\in\nu\\
			\false &\text{if }a\notin\nu
		\end{cases}\\
		\neg a[\nu]&=
		\begin{cases}
			\false &\text{if }a\in\nu\\
			\true &\text{if }a\notin\nu
		\end{cases}\\
		(\X\psi_1)[\nu]&= \psi_1\\
		\psi[\nu] &= \psi \text{ for any other $\psi$}
	\end{align*}
\end{minipage}
Note that after unfolding, a formula becomes a positive Boolean combination over literals (atomic propositions and their negations) and $\X$-formulae.
The resulting formula is LTL-equivalent to the original formula. 
The formulae of the form $\Gf{\bowtie p}_\ext\psi$ have ``dummy’’ unfolding; they are dealt with in a special way later.
Combined with unfolding, the ``next step''-operator then preserves and reflects satisfaction on the given word:

\begin{lemma}\label{lem:unfolding}
For every word $w$ and $\fltlfragment$ formula $\varphi$, we have 
$w\models \varphi$ if and only if $\suffix{w}{1} \models (\un(\varphi))[w[0]]$.
\end{lemma}
The construction of $\A$ proceeds in several steps. 
We first construct a ``master'' transition system, which monitors the formula and transforms it in each step to always keep exactly the formula that needs to be satisfied at the moment. 
However, this can only deal with properties whose satisfaction has a finite witness, e.g. $\F a$. 
Therefore we construct a set of ``slave'' automata, which check whether ``infinitary'' properties (with no finite witness), e.g., $\F\G a$, hold or not. 
They pass this information to the master, who decides on acceptance of the word.

\subsection{Construction of master transition system $\T$}

We define a LTS $\T=(Q,\varphi,\delta^{\T})$  over $2^{Ap}$ by letting
$Q$ be the set
 of positive Boolean functions\footnote{We use Boolean functions, i.e. classes of propositionally equivalent formulae, to obtain a finite state space. To avoid clutter, when referring to such a Boolean function, we use some formula representing the respective equivalence class.
 	The choice of the representing formula is not relevant since, for all operations we use, the propositional equivalence is a congruence, see
 \ifx\techreport\undefined\cite{techreport}\else Appendix~\ref{app:bool}\fi. Note that, in particular, $\true,\false\in Q$. } over $\sf$, by letting
 $\varphi$ be the initial state, and by letting the transition function $\delta^\T$, for every $\nu\subseteq Ap$ and $\psi \in Q$, contain 
 $\psi \tran\nu (\un(\psi))[\nu]$.

The master automaton 
keeps the property that is still required up to date:

\begin{lemma}[Local (finitary) correctness of master LTS]\label{lem:master-local}
	Let $w$ be a word and $\mathcal \T(w)=\varphi_0\varphi_1\cdots$ the corresponding run. Then for all $n\in\mathbb N$, we have $w\models\varphi$ if and only if $\suffix{w}{n} \models\varphi_n$.
\end{lemma}

\begin{example}
The formula $\varphi=a\wedge \X(b\U a)$ yields a master LTS depicted below.
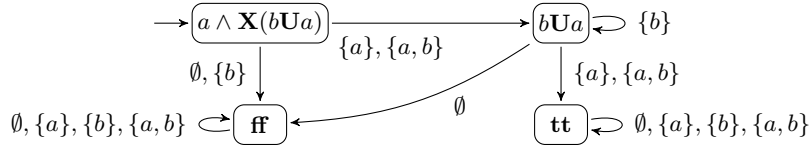
\begin{center}		
		\begin{tikzpicture}[x=4cm,y=1.5cm,font=\footnotesize,initial text=,outer sep=0.5mm]
		\tikzstyle{acc}=[double]
		\node[state,initial] (i) at (0,0) {$a\wedge \X(b\U a)$};
		\node[state] (t) at (1,-0.9) {$\true$};
		\node[state] (f) at (0,-0.9) {$\false$};
		\node[state] (u) at (1,0) {$b\U a$};

		\path[->] 
		(i) edge node[left]{$\emptyset,\{b\}$} (f)
		    edge node[pos=0.3,below]{$\{a\},\{a,b\}$} (u)
		(u) edge[loop right, max distance=8mm,in=-10,out=10,looseness=10] node[right]{$\{b\}$} (u)
		    edge[bend left=15] node[below,pos=0.3] {$\emptyset$} (f)
		    edge node[right] {$\{a\},\{a,b\}$} (t)
		(t) edge[loop right, max distance=8mm,in=-10,out=10,looseness=10] node[right]{$\emptyset,\{a\},\{b\},\{a,b\}$} (t)   
		(f) edge[loop left, max distance=8mm,in=170,out=190,looseness=10] node[left]{$\emptyset,\{a\},\{b\},\{a,b\}$} (f)   
		;
		\end{tikzpicture}
\end{center}
\end{example}

One can observe that for an fLTL formula $\varphi$ with no $\G$- and $\Gf{\bowtie p}_\ext$-operators,
we have $w\models\varphi$ if{}f the state $\true$ is reached while reading $w$.
However, for formulae with $\G$-operators (and thus without finite witnesses in general), this claim no longer holds.
To check such behaviour we construct auxiliary ``slave'' automata.

\subsection{Construction of slave transition systems $\S(\xi)$}

We define a LTS $\S(\xi)=(Q,\xi,\delta^{\S})$ over $2^{Ap}$
with the same state space as $\T$ and the initial state $\xi \in Q$.
Furthermore, we call a state $\psi$ a \emph{sink}, written $\psi\in\sink$, iff for all $\nu\subseteq Ap$ we have $\psi[\nu]=\psi$.
Finally, 
the transition relation $\delta^{\S}$, for every $\nu\subseteq Ap$ and $\psi\in Q \setminus \sink$, contains 
$
\psi\tran{\nu}\psi[\nu] 
$.

\begin{example}\label{ex:slave}
The slave LTS for the formula $\xi=a\vee b\vee\X(b\wedge \G \F a)$ has a structure depicted in the following diagram:
	
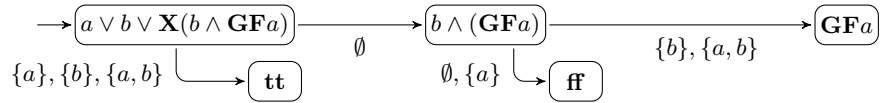
\begin{center}	
		\begin{tikzpicture}[x=4cm,y=1.5cm,font=\footnotesize,initial text=,outer sep=0.5mm]
		\tikzstyle{acc}=[double]
		\node[state,initial] (i) at (0,0) {$a\vee b\vee\X(b\wedge\G\F a)$};
		\node[state] (t) at (0.3,-0.5) {$\true$};
		\node[state] (a) at (1,0) {$b\wedge(\G\F a)$};
		\node[state] (f) at (1.3,-0.5) {$\false$};
		\node[state] (g) at (2.2,0) {$\G\F a$};
		
		\draw[->,rounded corners] (i.-110) |- (t) node[pos=0.4,left]{$\{a\},\{b\},\{a,b\}$};
		\draw[->,rounded corners] (i) edge node[below]{$\emptyset$} (a)
		(a) edge node[pos=0.6, below] {$\{b\},\{a,b\}$} (g)
		(a.-40) |- (f) node[pos=0.4,left] {$\emptyset,\{a\}$} 
		;
		\end{tikzpicture}
	\end{center}
Note that we do not unfold any inner \F- and \G-formulae.
Observe that if we start reading $w$ at the $i$th position and end up in $\true$, we have $\suffix w i\models \xi$.
Similarly, if we end up in $\false$ we have $\suffix w i\not\models \xi$.
This way we can monitor for which position $\xi$ holds and will be able to determine if it holds, for instance, infinitely often. 
But what about when we end up in $\G\F a$? Intuitively, this state is accepting or rejecting depending on whether $\G\F a$ holds or not.
Since this cannot be checked in finite time, we delegate this task to yet another slave, now responsible for $\G\F a$.
Thus instead of deciding whether $\G\F a$ holds, we may use it as an \emph{assumption} in the automaton for $\xi$ and let the automaton for $\G\F a$ check whether the assumption turns out correct.
\end{example}

Let $\rec := \fsf \cup \gsf \cup \gsf^{\bowtie}$. This is the set of subformulas that are potentially difficult to check in finite time.
Subsets of $\rec$ can be used as assumptions to prove other assumptions and in the end also the acceptance.
Given a set of formulae $\Psi$ and a formula $\psi$, we say that $\Psi$ 
\emph{(propositionally) proves} $\psi$, written $\Psi\proves\psi$, if $\psi$ can be deduced from formulae in $\Psi$ using only propositional reasoning (for a formal definition see
\ifx\techreport\undefined\cite{techreport}\else Appendix~\ref{app:bool}\fi). 
So, for instance, $\{\G\F a\}\proves\G\F a \vee \G b$, but $\G\F a\not\proves\F a$. 

The following is the ideal assumption set we would like our automaton to identify.
For a fixed word $w$, we denote by $\R(w)$ the set
\[
\{ \F \xi \in \fsf \mid w \models \G \F \xi  \} 
  \cup \{ \G \xi \in \gsf \mid w \models \F \G \xi  \}
  \cup \{ \Gf{\bowtie p}_\ext \xi \in \gsf^{\bowtie} \mid w \models \Gf{\bowtie p}_\ext \xi  \}
\]
 of formulae in $\rec$ eventually always satisfied on $w$.
The slave LTS is useful for recognizing whether its respective formula $\xi$ holds infinitely often, almost always, or with the given frequency.
Intuitively, it reduces this problem for a given formula to the problems for its subformulas in $\rec$:

\begin{lemma}[Correctness of slave LTS]\label{lem:slave-promises} 	
Let us fix $\xi \in \sf$ and a word~$w$. For any $\R \in \rec$, we denote by $\sat(\R)$ the set $\{i\in\Nset\mid \exists j\geq i:\R \proves \S(\xi)(\infix wij) \}$. %
Then for any $\underline{\R},\overline{\R} \subseteq \rec$ such that $\underline{\R} \subseteq \R(w) \subseteq \overline{\R}$, we have
	\begin{eqnarray}
	\sat(\underline{\R}) \text{ is infinite} \implies & w \models \G\F\xi \implies & \sat(\overline{\R}) \text{ is infinite}\\
	\Nset \setminus \sat(\underline{\R}) \text{ is finite} \implies & w \models \F\G\xi \implies & \Nset \setminus \sat(\overline{\R}) \text{ is finite}\\
    \lrExt(\big(\idf_{i\in \sat(\underline{\R})}\big)_{i=0}^\infty) \bowtie p \implies & w \models \Gf{\bowtie p}_\ext \xi \implies & \lrExt(\big(\idf_{i\in\sat(\overline{\R})}\big)_{i=0}^\infty) \bowtie p \qquad
\end{eqnarray}
\end{lemma}
Before we put the slaves together to determine $\R(w)$, we define \emph{slave automata}.
In order to express the constraints from Lemma~\ref{lem:slave-promises} as acceptance conditions, we need to transform the underlying LTS.
Intuitively, we replace quantification over various starting positions for runs by a subset construction.
This means that in each step we put a \emph{token} to the initial state and move all previously present tokens to their successor states.

\mypara{B\"uchi}
For a formula $\F \xi \in \fsf$, its
slave LTS $\S(\xi)=(Q,\xi,\delta^{\S})$, and $\R \subseteq \rec$,
we define a B\"uchi automaton $\SGF(\xi,\R)=(2^{Q},\{\xi\},\delta%
)$ over $2^{Ap}$ by setting
\[\Psi\tran{\nu}\{\delta^{\S}(\psi,\nu)\mid\psi\in\Psi\setminus \sink\}\cup\{\xi\} \qquad \qquad \text{for every $\nu\subseteq Ap$} \]
and the B\"uchi acceptance condition 
$\accInfinite(\{\Psi\subseteq Q\mid \exists \psi\in\Psi \cap \sink: \R\proves\psi\})$.

\smallskip

In other words, the automaton accepts if infinitely often a token ends up in an \emph{accepting sink}, i.e., element of $\sink$ that is provable from $\R$. 
For Example~\ref{ex:slave}, depending on whether we assume $\G\F a\in\R$ or not, the accepting sinks are $\true$ and $\G\F a$, or only $\true$, respectively.

\mypara{Co-B\"uchi}
For a formula $\G \xi \in \gsf$, its slave LTS $\S(\xi)=(Q,\xi,\delta^{\S})$ and $\R \subseteq \rec$, we define a co-B\"uchi automaton $\SFG(\xi,\R)=(2^{Q},\{\xi\},\delta%
)$ over $2^{Ap}$ with the same LTS as above.
It differs from the B\"uchi automaton only by having 
a co-B\"uchi acceptance condition $\accFinite(\{\Psi\subseteq Q\mid \exists \psi \in \Psi\cap \sink : \R \notproves \psi\})$.

\mypara{Mean-payoff}
For a formula $\Gf{\bowtie p}_\ext \xi\in \gsf^{\bowtie}$, its slave LTS $\S(\xi)=(Q,\xi,\delta^{\S})$, and $\R \subseteq \rec$ we define a \emph{mean-payoff automaton} 
 $\SGf(\xi,\R)= (|Q|^{Q},\idf_{\xi},\delta%
 )$ over $2^{Ap}$ so that for every $\nu\subseteq Ap$, we have $f\tran{\nu}f'$ where 
\[
f'(\psi') = \idf_\xi(\psi')+\sum_{\delta^\S(\psi,\nu) = \psi'} f(\psi).
\]
Intuitively, we always count the number of tokens in each state. When a step is taken, all tokens moving to a state are summed up and, moreover, one token is added to the initial state. 
Since the slave LTS is acyclic the number of tokens in each state is bounded.

Finally, the acceptance condition is $\accMP[\ext]{\bowtie p}{\singlereward(\R)}$ where the function $\singlereward(\R)$ assigns to every state $f$ the reward:
\[
\sum_{\psi\in\sink,\R \proves \psi} f(\psi).
\]
Each state thus has a reward that is the number of tokens in accepting sinks. 
Note that each token either causes a reward 1 once per its life-time when it reaches an accepting sink, or never causes any reward in the case when it never reaches any accepting state.

\begin{lemma}[Correctness of slave automata]\label{lem:slaves-acc}
Let $\xi \in \sf$, $w$, and $\underline{\R},\overline{\R} \subseteq \rec$ be such that $\underline{\R} \subseteq \R(w) \subseteq \overline{\R}$. Then 
	\begin{eqnarray}
		w \in \L(\SGF(\xi,\underline{\R}))  \implies & w \models \G\F \xi & \implies w \in \L(\SGF(\xi,\overline{\R}))\\
		w \in \L(\SFG(\xi,\underline{\R}))  \implies & w \models \F\G \xi & \implies w \in \L(\SFG(\xi,\overline{\R}))\\
		w \in \L(\SGf(\xi,\underline{\R}))  \implies & w \models \Gf{\bowtie p}_\ext \xi & \implies w \in \L(\SGf(\xi,\overline{\R}))
	\end{eqnarray}
\end{lemma}

\vspace{-0.9em}

\subsection{Product of slave automata}

Observe that the LTS of slave automata never depend on the assumptions $\R$. Let $\S_1,\ldots,\S_n$ be the LTS of automata for elements of $\rec = \{\xi_1,\ldots,\xi_n\}$. Further, given $\R\subseteq\rec$, let $\acc_i(\R)$ be the acceptance condition for the slave automaton for $\xi_i$ with assumptions $\R$.

We define $\P$ to be the LTS product $\S_1 \times \cdots \times \S_n$.
The slaves run independently in parallel. 
For $\R \subseteq\rec$, we define the acceptance condition for the product\footnote{An acceptance condition of an automaton is defined to hold on a run of the automata product if it holds on the projection of the run to this automaton. We can still write this as a standard acceptance condition. Indeed, for instance, a B\"uchi condition for the first automaton given by $F\subseteq Q$ is a B\"uchi condition on the product given by $\{(q_1,q_2,\ldots,q_n)\mid q_1\in F, q_2,\ldots,q_n\in Q\}$.}
\[
\acc(\R) = \bigwedge_{\xi_i\in\R} \acc_i(\R)
\]
and $\P(\R)$ denotes the LTS $\P$ endowed with the acceptance condition $\acc(\R)$. 
Note that $\acc(\R)$ checks that $\R$ is satisfied when each slave assumes $\R$.

\begin{lemma}[Correctness of slave product]\label{lem:slaves-product}
	For $w$ and $\R \subseteq \rec$, we have
\begin{description}
\item[(soundness)]  whenever
        $w \in \L(\P({\R}))$ then $\R\subseteq\R(w)$;
\item[(completeness)]    $w \in \L(\P({\R(w)}))$.    
\end{description}
\end{lemma}
Intuitively, soundness means that whatever set of assumptions we prove with $\P$ it is also satisfied on the word. Note that the first line can be written as
\[
w \in \L(\P({\R}))  \implies w \models \bigwedge_{\F\xi \in \R} \G\F \xi \land \bigwedge_{\G\xi \in \R} \F\G \xi \land \bigwedge_{\Gf{\bowtie p}_\ext \xi \in \R} \Gf{\bowtie p}_\ext \xi
\]
Completeness means that for every word the set of all satisfied assumptions can be proven by the automaton.

\subsection{The final automaton: product of slaves and master}

Finally, we define the generalized Rabin mean-payoff automaton $\A$ to have the 
LTS $\T \times \P$ 
and the acceptance condition
\(
\bigvee_{\R\subseteq\rec} \accT{\R} \wedge \accP{\R}
\)
where
\[
\accT{\R} = \accFinite\Big(\Big\{\big(\psi,(\Psi_\xi)_{\xi\in \rec}\big)\ \Big|\ 
\R\cup \bigcup_{\G\xi\in\R} \Psi_\xi[(\rec\setminus\R)/\false] \not\proves \psi\Big\}\Big)\,
\]
eventually prohibits states where the current formula of the master $\psi$ is not proved by the assumptions and by all tokens of the slaves for $\G\xi\in\R$. 
Here $\Psi[X/\false]$ denotes the set of formulae of $\Psi$ where each element of $X$ in the Boolean combination is replaced by $\false$. For instance, $\{a\vee\F a\}[\{a\}/\false]=\false\vee\F a=\F a$. (For formal definition, see
\ifx\techreport\undefined\cite{techreport}\else Appendix~\ref{app:bool}\fi.)
We illustrate how the information from the slaves in this form helps to decide whether the master formula holds or not.

\begin{example}
Consider $\varphi=\G(\X a\vee\G\X b)$, and its respective master transition system as depicted below:
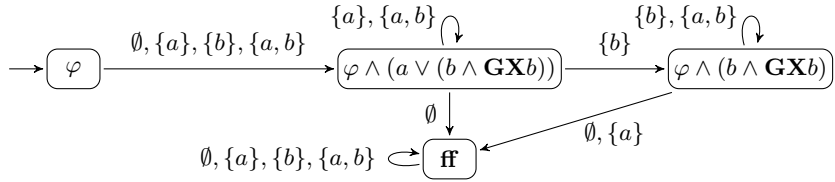
\begin{center}
		\begin{tikzpicture}[x=5cm,y=1.5cm,font=\footnotesize,initial text=,outer sep=0.5mm]
		\tikzstyle{acc}=[double]
		\node[state,initial] (x) at (-1,0) {$\varphi$};
		\node[state] (i) at (0,0) {$\varphi\wedge(a\vee(b\wedge\G\X b))$};
		\node[state] (f) at (0,-0.8) {$\false$};
		\node[state] (u) at (0.8,0) {$\varphi\wedge(b\wedge\G\X b)$};

		\path[->] 
		(x) edge node[above]{$\emptyset,\{a\},\{b\},\{a,b\}$} (i)
		(i) edge node[left]{$\emptyset$} (f)
		(i) edge[loop above] node[left]{$\{a\},\{a,b\}$} ()
		    edge node[above]{$\{b\}$} (u)
		(u) edge[loop above] node[left]{$\{b\},\{a,b\}$} (u)
		    edge node[below,pos=0.3] {$\emptyset,\{a\}$} (f)
		(f) edge[loop left, max distance=8mm,in=170,out=190,looseness=10] node[left]{$\emptyset,\{a\},\{b\},\{a,b\}$} (f)   
		;
		\end{tikzpicture}
\end{center}
Assume we enter the second state and stay there forever, e.g., under words $\{a\}^\omega$ or $\{a,b\}^\omega$.
How do we show that $\varphi\wedge(a\vee(b\wedge\G\X b))$ holds?
For the first conjunct, we obviously have $\R\proves\varphi$ for all $\R$ containing $\varphi$.
However, the second conjunct is more difficult to prove.

One option is that we have $\G\X b\in\R$ and want to prove the second disjunct. To this end, we also need to prove $b$.
We can see that if $\G\X b$ holds then in its slave for $\X b$, 
there is always a token in the state $b$, 
which is eventually always guaranteed to hold.
This illustrates why we need the tokens of the $\G$-slaves for proving the master formula.

The other option is that $\G\X b$ is not in $\R$, and so we need to prove the first disjunct.
However, from the slave for $\G(\X a\vee\G\X b)$ we eventually always get only the tokens $\X a\vee\G\X b$, $a\vee\G\X b$, and $\true$. 
None of them can prove $a\vee (b\wedge\G\X b)$.
However, since the slave does not rely on the assumption $\G\X b$, we may safely assume it not to hold here.
Therefore, we can substitute $\false$ for $\G\X b$ and after the substitution the tokens turn into $\X a$, $a$, and $\true$. 
The second one is then trivially sufficient to prove the first disjunct. 
\end{example}

\begin{proposition}[Soundness]\label{prop:A-sound}
	If $w \in \L(\A)$, then $w \models \varphi$.
\end{proposition}
The key proof idea is that for the slaves of $\G$-formulae in $\R$, all the tokens eventually always hold true.
Since also the assumptions hold true so does the conclusion $\psi$. By Lemma~\ref{lem:master-local}, $\varphi$ holds true, too.

\begin{proposition}[Completeness]\label{prop:A-complete}
	If $w \models \varphi$, then $w \in \L(\A)$.
\end{proposition}
The key idea is that subformulas generated in the master from $\G$-formulae closely correspond to their slaves' tokens.
Further, observe that for an $\F$-formula $\chi$, its unfolding is a disjunction of $\chi$ and other formulae.
Therefore, it is sufficient to prove $\chi$, which can be done directly from $\R$.
Similarly, for $\Gf{\bowtie p}_\ext$-formula $\chi$, its unfolding is just $\chi$ and is thus also provable directly from $\R$.

\mypara{Complexity}
Since the number of Boolean functions over a set of size $n$ is $2^{2^n}$, the size of each automaton is bounded by $2^{2^{|\sf|}}$, i.e., doubly exponential in the length of the formula.
Their product is thus still doubly exponential.
Finally, the acceptance condition is polynomial for each fixed $\R\subseteq\rec$. 
Since the whole condition is a disjunction over all possible values of $\R$, it is exponential in the size of the formula, which finishes the proof of Theorem~\ref{thm:auto}.

\section{Verifying strongly connected MDPs against generalized B\"uchi mean-payoff automata}
\label{sec:mean-payoff}
\begin{figure}[t]
\begin{align}
\sum\nolimits_{a\in A} x_{i,a} & =  1 & \text{for all $1\le i\le n$}
\label{eq:xssum}\\
\sum\nolimits_{a\in A} x_{i,a}\cdot \delta(a)(s) & = 
\sum\nolimits_{a\in \act{s}} x_{i,a} & \text{for all $s\in S$ and $1{\le} i{\le} n$}
\label{eq:xa}\\
\sum\nolimits_{s\in S, a\in \act{s}} x_{i,a}\cdot\singlereward_j(s) & \bowtie  v_j &
\text{ for all $1{\le} j {\le} m$ and $1{\le} i{\le} n$}
\label{eq:rew-inf}\\
\sum\nolimits_{s\in S, a\in \act{s}} x_{i,a}\cdot\singlerewardalt_i(s) & \bowtie  u_i &
\text{ for all $1{\le} i {\le} n$}
\label{eq:rew-sup}
\end{align}
\myspace\myspace
\caption{Linear constraints $L$ of Proposition~\ref{thm:mp-main}}
\label{system-L}
\end{figure}
Theorem~\ref{thm:mec} can be obtained from the following proposition.

\begin{proposition}\label{thm:mp-main}
 Let $\mdp=(S,A,\mathit{Act},\delta,\inits)$ be a strongly connected MDP, and $\acc$ an acceptance condition over $S$ given by:
 \[
  \bigwedge\nolimits_{i=1}^k\Inf(S_i) \quad\wedge\quad \bigwedge\nolimits_{i=1}^m \accMP[\inf]{\bowtie v_i}{\singlereward_i} \quad\wedge\quad \bigwedge\nolimits_{i=1}^n \accMP[\sup]{\bowtie u_i}{\singlerewardalt_i)}
 \]
 The constraints from Figure~\ref{system-L} have a non-negative solution if and only if
 there is a strategy $\sigma$ and %
 a set of runs $R$ of non-zero probability such that $\acc$ holds true on all $\omega\in R$.
Moreover, $\sigma$ and $R$ can be chosen so that $R$ has probability~$1$.
\end{proposition}
Intuitively, variables $x_{i,a}$ describe the frequencies of using action $a$. Equation (\ref{eq:xa}) is Kirchhof's law of flow. Equation (\ref{eq:rew-inf}) says the inferior limits must be satisfied by all flows, while Equation (\ref{eq:rew-sup}) says that the $i$th limit superior has its own dedicated $i$th flow.
Note that $L$ does not dependent on the initial state $\inits$.

\begin{proof}[Sketch]
Existing results for multi-objective mean payoff MDPs would only allow to establish the proposition in absence of supremum limits,
and so we need to extend and combine results of several works to prove the proposition.
In the direction $\Rightarrow$, \cite[Corollary 12]{lics15} gives a strategy $\sigma_i$ for every $i$ such that for almost every run $s_0a_0s_1a_1\ldots$ we have $\lrInf((\idf_{a_t = a})_{t=0}^\infty) = x_{i,a}$, and in fact
the corresponding limit exists. Hence, for the number $p=\sum_{s\in S, a\in \act{s}} \singlereward(s)\cdot x_{i,a}$ the predicates $\accMP[\inf]{\ge p}{\singlereward}$ and $\accMP[\sup]{\ge p}{\singlereward}$ almost surely holds, for any reward function $\singlereward$.
Hence, our constraints ensure that $\sigma_i$ satisfies $\accMP[\inf]{\bowtie v_j}{\singlereward_j}$ for all $j$, and $\accMP[\sup]{\bowtie u_i}{\singlerewardalt_i}$. 
Moreover, $\sigma_i$ is guaranteed to visit every state of $\mdp$ infinitely often almost surely. The strategy $\sigma$ is then constructed to take these strategies $\sigma_i,1\leq i\leq n$ in turn and mimic each one of them for longer and longer
periods.

For the direction $\Leftarrow$, we combine the ideas of \cite{lics15,BBC+14,BCFK13} and select solutions to $x_{i,a}$ from ``frequencies'' of actions under the strategy $\sigma$. 

\end{proof}

\vspace*{-4mm}
\vspace*{-1mm}
\section{Conclusions}
\label{sec:conclusions}

We have given an algorithm for computing the optimal probability of satisfying an $\fltlfragment$ formula
in an MDP. The proof relies on a decomposition of the formula into master and slave automata, and on solving a mean-payoff
problem in a product MDP.
The obvious next step is to extend the algorithm so that it can handle arbitrary formulae of fLTL. This appears to be
a major task, since our present construction relies
on acyclicity of slave LTS, a property which is not satisfied for unrestricted formulae~\cite{cav14}.
Indeed, since $\Gf{\bowtie p}$-slaves count the number of tokens in each state, this property ensures a bounded number of tokens and thus finiteness of the slave automata.

\mypara{Acknowledgments.}
{\footnotesize This work is partly supported by the German Research Council (DFG) as part of the Transregional Collaborative Research Center AVACS (SFB/TR 14), by the Czech Science Foundation under grant agreement P202/12/G061, by the EU 7th Framework Programme under grant agreement no. 295261 (MEALS) and 318490 (SENSATION), by the CDZ project 1023 (CAP), by the CAS/SAFEA International Partnership Program for Creative Research Teams, by the EPSRC grant EP/M023656/1, by the People Programme (Marie Curie Actions) of the European Union’s
Seventh Framework Programme (FP7/2007–2013) REA Grant No 291734, by the Austrian Science Fund (FWF) S11407-N23 (RiSE/SHiNE), and by the ERC Start Grant (279307: Graph Games). Vojt\v{e}ch Forejt is also affiliated with FI MU, Brno, Czech Republic.}

\vspace*{-3mm}
\bibliographystyle{abbrv}
\bibliography{bib}

\ifx\techreport\iundefined\relax\else
\newpage
\appendix

\section{Propositional reasoning}\label{app:bool}

Intuitively, given a set of formulae $\Phi$ and a formula $\psi$, we say that $\Phi$ 
propositionally proves $\psi$ if $\psi$ can be deduced from formulae in $\Phi$ using only 
propositional reasoning. So, for instance, $\G a$ propositionally implies $\G a \vee \G b$,
but $\G a$ does not propositionally imply $\F a$. 

\begin{definition}[Propositional implication and equivalence]
A formula of fLTL is {\em non-Boolean} if it is not a conjunction or a disjunction
(i.e., if the root of its syntax tree is not $\wedge$ or $\vee$). 
The set of non-Boolean formulae of fLTL over $Ap$ is denoted by ${\it NB}(Ap)$.
A {\em propositional assignment}, or just an {\em assignment}, is a mapping 
$\mathit{Ass} \colon {\it NB}(Ap) \rightarrow \{0, 1\}$. Given $\varphi \in {\it NB}(Ap)$, 
we write $\mathit{Ass} \models_P \varphi$ if{}f $\mathit{Ass}(\varphi) = 1$, and extend
the relation $\models_P$ to arbitrary formulae by: 

\[\begin{array}[t]{lcl}
  \mathit{Ass}\models_P \varphi \wedge \psi & \mbox{ if{}f } & \mathit{Ass} \models_P \varphi \text{ and } \mathit{Ass} \models_P \psi \\
  \mathit{Ass} \models_P \varphi \vee \psi & \mbox{ if{}f } & \mathit{Ass} \models_P \varphi \text{ or } \mathit{Ass} \models_P \psi \\
\end{array}\]

We say that a set $\Phi$ of fLTL formulae \emph{propositionally proves} an fLTL formula $\psi$, written $\Phi\proves\psi$,
if for every assignment $\mathit{Ass}$, $\mathit{Ass}\models_P \bigwedge\Phi$ implies $\mathit{Ass}\models_P \psi$.

Finally, fLTL formulae $\varphi$ and $\psi$ are {\em propositionally equivalent}, denoted by $\varphi \equiv_P \psi$, 
if $\{\varphi\} \models_P \psi$ and $\{\psi\} \models_P \varphi$. We denote by $[\varphi]_P$ the equivalence class of $\varphi$ 
under the equivalence relation $\equiv_P$.
\end{definition}

Observe that $\varphi \equiv_P \psi$ implies that $\varphi$ and $\psi$ are equivalent also as {fLTL} formulae, i.e., for all words $w$, we have $w\models\varphi$ iff $w\models\psi$. Using the same reasoning, 
\begin{align}
w\models\bigwedge\Phi\text{ with }\Phi\proves\psi\text{ imply }w\models\psi\,.\label{eq:transitivity}
\end{align}

\begin{definition}[Propositional substitution]
Let $\psi,\chi$ be fLTL formulae and $\Psi\subseteq{\it NB}(Ap)$. %
The formula $\psi[\Psi/ \chi]_P$ is inductively defined as follows:
\begin{itemize}
\item If $\psi = \psi_1 \wedge \psi_2$ then $\psi[\Psi/ \chi]_P = \psi_1[\Psi/\chi]_P \wedge \psi_2[\Psi/ \chi]_P$.
\item If $\psi = \psi_1 \vee \psi_2$ then $\psi[\Psi/ \chi]_P = \psi_1[\Psi/\chi]_P \vee \psi_2[\Psi/ \chi]_P$.
\item If $\psi$ is a non-Boolean formula and $\psi \in \Psi$ then $\psi[\Psi/ \chi]_P = \chi$, else $\psi[\Psi/ \chi]_P = \psi$.
\end{itemize}
\end{definition}

The following lemma allows us to work with formulae as Boolean functions over $\mathit{NB}(Ap)$, i.e.,
as representatives of their propositional equivalence classes.

\begin{alemma}
For every formula $\varphi$ and every %
letter $\nu \in 2^{Ap}$,
if $\varphi_1 \equiv_P \varphi_2$ then $\un(\varphi_1))[\nu] \equiv_P \un(\varphi_2))[\nu]$.
\end{alemma}
\begin{proof}
Observe that every formula $\varphi$ is a positive Boolean combination (i.e., built from conjunctions and disjunctions) of non-Boolean formulae. 
Since $\un$ and $(\cdot)[\nu]$ both distribute over $\wedge $ and $\vee$, the formula
$\un(\varphi)[\nu]$ is obtained by applying a simultaneous substitution to the non-Boolean formulae.
(For example, a non-Boolean formula $\G\psi$ is substituted by $\un(\psi)[\nu] \wedge \G\psi$.) 
Let $\varphi[S]$ be the result of the substitution. 

Consider two equivalent formulae $\varphi_1 \equiv_P \varphi_2$. Since we apply the same substitution to 
both sides, the substitution lemma of propositional logic guarantees $\varphi_1[S] \equiv_P \varphi_2[S]$.
\end{proof}

\section{Proofs of Section~\ref{sec:automata}}

\begin{reflemma}{lem:unfolding}
For every word $w$ and fLTL formula, we have $w\models \varphi\iff \suffix{w}{1} \models (\un(\varphi))[w[0]]$.
\end{reflemma}

\begin{proof}
Denote $w=\nu v$ where $\nu\subseteq Ap$. We proceed by a straightforward structural induction on $\varphi$. We focus on three representative cases.
\begin{itemize}
\item $\varphi = a$. Then
\[
\begin{array}{llr}
             & \nu v \models a   \\[0.1cm] 
\iff & a \in \nu & \text{(semantics of LTL)}  \\[0.1cm]
\iff & \un(a)[\nu] = \true & \text{(def. of $\un$ and $[\nu]$)} \\[0.1cm] 
\iff & v \models \un(a)[\nu] & \text{(semantics of LTL)} \\[0.1cm] 
\end{array}
\]

\item $\varphi = \F \psi$. Then
\[\begin{array}{llr}
              & \nu v \models \F \psi \\[0.1cm]
\iff & \nu v \models (\X\F \psi) \vee \psi & \text{($\F \psi \equiv \X\F \psi \vee \psi$)}\\[0.1cm]
\iff & v \models \F \psi \text{ or }  \nu v \models \psi & \text{(semantics of LTL)}\\[0.1cm]
\iff & v \models \F \psi \text{ or }  v \models \un(\psi)[\nu] & \text{(ind. hyp.)}\\[0.1cm]
\iff & v \models \F \psi \vee \un(\psi)[\nu] & \text{(semantics of LTL)}\\[0.1cm]
\iff & v \models \un(\F\psi)[\nu] & \text{(def. of $\un$)}
\end{array}\]

\item $\varphi = \Gf{\bowtie p}_{\inf} \psi$. Then
\[\begin{array}{llr}
              & \nu v \models \Gf{\bowtie p}_{\inf} \psi \\[0.1cm]
\iff & \displaystyle \liminf_{i\to \infty} \frac{1}{i} \Big(\idf_{\nu\models\varphi}+\sum_{j=0}^{i-2} \idf_{\suffix{v}{j}\models\varphi}\Big) & \text{(semantics of LTL)}\\[0.1cm]
\iff & \displaystyle \lim_{i\to \infty} \frac{1}{i} \idf_{\nu\models\varphi}+\liminf_{i\to \infty} \frac{1}{i} \sum_{j=0}^{i-2} \idf_{\suffix{v}{j}\models\varphi}\Big) &\\[0.1cm]
\iff & \displaystyle 0+\liminf_{i\to \infty} \frac{1}{i} \sum_{j=0}^{i-1} \idf_{\suffix{v}{j}\models\varphi} &\\[0.1cm]
\iff & v \models \Gf{\bowtie p}_{\inf} \psi  & \text{(semantics of LTL)}\\[0.1cm]
\iff & v \models \un(\Gf{\bowtie p}_{\inf} \psi )[\nu] & \text{(def. of $\un$)}
\end{array}\]
\end{itemize}
\qed
\end{proof}

\begin{reflemma}{lem:master-local}
Let $w$ be a word and $\mathcal \T(w)=\varphi_0\varphi_1\cdots$ the corresponding run. Then for all $n\in\mathbb N$, we have $w\models\varphi$ if and only if $\suffix{w}{n} \models\varphi_n$.
\end{reflemma}
\begin{proof}
We proceed by induction on $n$. 
For $n=0$, we conclude by $\varphi_0=\varphi$.
Let now $n\geq1$ and denote $w=u\,\nu\,v$ where $\nu\subseteq AP$ and $v=\suffix{w}{n}$.
Then we have 

\[\begin{array}{llr}
     & v \models \varphi_n \qquad\qquad\qquad\\
\iff & v \models \un(\varphi_{n-1})[\nu] & \text{(def. of $\delta^\T$)} \\
\iff & \nu\,v \models \varphi_{n-1} & \text{(Lemma~\ref{lem:unfolding})} \\
\iff & u\,\nu\,v\models \varphi & \text{(ind. hyp.)} 
\end{array}\]
\qed
\end{proof}

\begin{definition}
The \emph{threshold} $T(w)$ of a word $w$ is the smallest $T\in\Nset$ such that for all $t\geq T$
\begin{itemize}
\item for all $\psi\in\R(w)$, we have $\suffix wt\models\psi$,\footnote{This condition is actually non-trivial only for $\G$-formulae, other formulae of $\R(w)$ hold at all positions.}
\item for all $\psi\in\rec\setminus\R(w)$, we have $\suffix wt\not\models\psi$.
\end{itemize}
\end{definition}

Then we have $\suffix w {T(w)}\models\rho$ for every $\rho\in\R(w)$ (all $\G$-formulae that will ever hold do hold already) and $\suffix w {T(w)}\not\models\rho$ for every $\rho\in\rec\setminus\R(w)$ (none of the $\F$-formulae that hold only finitely often holds any more).

\begin{alemma}\label{lem:threshold}
For every word $w$ and $t\geq T(w)$, we have that  $\suffix w t\models \xi$ if{}f $\exists t'%
:\R(w)\cap\sf(\xi) \proves \S(\xi)(\infix w t {t'})$.
\end{alemma}%
\begin{proof}
By similar arguments as in Lemma~\ref{lem:master-local}, we get that for the run of the slave $\S(\xi)(\suffix wt) = \xi_t \xi_{t+1} \cdots$ we have $\suffix wt \models \xi \iff \suffix w{t'} \models \xi_{t'}$. Indeed, not unfolding elements of $\rec$ is here equivalent to not unfolding them since for every $\psi\in\rec$ we have $\suffix wu\models\psi$ iff $\suffix wu\models \X\psi$, for all $u\geq T$.
Moreover, when reaching the sink at time $t'$, we know that $\xi_{t'}$ is a positive Boolean combination over $\rec(w) \cap \sf(\xi)$. 
Therefore, $\suffix wt \models \xi \iff \suffix w{t'} \models\xi_{t'}\iff \R(w)\models\xi_{t'}\iff \R(w)\cap\sf(\xi)\models\xi_{t'}$.
\qed
\end{proof}

\begin{reflemma}{lem:slave-promises} 	
Let us fix $\xi \in \sf$ and a word~$w$. For any $\R \in \rec$, we denote by $\sat(\R)$ the set $\{i\in\Nset\mid \exists j\geq i:\R \proves \S(\xi)(\infix wij) \}$. %
Then for any $\underline{\R},\overline{\R} \subseteq \rec$ such that $\underline{\R} \subseteq \R(w) \subseteq \overline{\R}$, we have
\begin{eqnarray}
\sat(\underline{\R}) \text{ is infinite} \implies & w \models \G\F\xi \implies & \sat(\overline{\R}) \text{ is infinite} \label{eq:sl1}\\
\Nset \setminus \sat(\underline{\R}) \text{ is finite} \implies & w \models \F\G\xi \implies & \Nset \setminus \sat(\overline{\R}) \text{ is finite} \label{eq:sl2} \\
\lrExt(\big(\idf_{\sat(\underline{\R})}(i)\big)_{i=0}^\infty) \bowtie p \implies & w \models \Gf{\bowtie p}_\ext \xi \implies & \lrExt(\big(\idf_{\sat(\overline{\R})}(i)\big)_{i=0}^\infty) \bowtie p \qquad\label{eq:sl3}
\end{eqnarray}
Moreover, the result holds also for $\underline{\R} \cap \sf(\xi) \subseteq \R(w) \cap \sf(\xi) \subseteq \overline{\R} \cap \sf(\xi)$. 
\end{reflemma}

\begin{proof}
For (\ref{eq:sl1}), let first $\sat(\underline{\R})$ be infinite. Then also $\sat'(\underline{\R}) := \{n\in\sat(\underline{\R}) \mid n\geq T(w)\}$ is infinite.
Therefore, infinitely many positions $i$ of $w$ satisfy $\exists j\geq i:\underline{\R} \proves \S(\xi)(\infix wij)$.
Observe that elements of $\rec$ are never under the scope of negation in $Q$, hence $\proves$ is monotonic w.r.t.\ adding assumptions from $\rec$.
Thus also infinitely many positions $i$ of $w$ satisfy $\exists j\geq i:\R(w) \proves \S(\xi)(\infix wij)$ and by Lemma~\ref{lem:threshold} also satisfy $\xi$.

Let now $w\models \G\F\xi$. Then $I:=\{i\in\Nset\mid i\geq T(w) \text{ and }\suffix w i \models \xi\}$ is infinite and by the lemma there are infinitely many positions $i$ of $w$ satisfying $\exists j\geq i:{\R} \proves \S(\xi)(\infix wij)$. 
By the monotonicity of $\proves$ above we can replace $\R$ by $\overline{\R}$.

Moreover, if we only assume $\underline{\R} \cap \sf(\xi) \subseteq \R(w) \cap \sf(\xi) \subseteq \overline{\R} \cap \sf(\xi)$ both statements remain valid.
Indeed, for every set $\R$ of formulae and formula reachable from $\xi$, $\R\proves\xi$ if{}f $\R\cap\sf(\xi)\proves\xi$ since the only non-Boolean formulae produced by $\xi[\cdot]$ are subformulas of $\xi$.

For (\ref{eq:sl2}), the argumentation is the same, replacing ``infinite'' and ``infinitely many'' by ``co-finite'' and ``almost all''.
For (\ref{eq:sl3}), the sequences can only differ in a finite prefix. 
Moreover, if we only assume $\underline{\R} \cap \sf(\xi) \subseteq \R(w) \cap \sf(\xi) \subseteq \overline{\R} \cap \sf(\xi)$,
apart from the finite prefix the sequence $\idf_{\sat(\underline{\R})}(i)$ is pointwise less or equal to $\idf_{\suffix wi\models\xi}$, 
which is again pointwise less or equal to $\idf_{\sat(\overline{\R})}(i)$.
\qed
\end{proof}

\begin{reflemma}{lem:slaves-acc}
Let $\xi \in \sf$, $w$, and $\underline{\R},\overline{\R} \subseteq \rec$ be such that $\underline{\R} \subseteq \R(w) \subseteq \overline{\R}$. Then 
\begin{eqnarray}
w \in \L(\SGF(\xi,\underline{\R}))  \implies & w \models \G\F \xi & \implies w \in \L(\SGF(\xi,\overline{\R})) \label{eq:sl4}\\
w \in \L(\SFG(\xi,\underline{\R}))  \implies & w \models \F\G \xi & \implies w \in \L(\SFG(\xi,\overline{\R})) \label{eq:sl5}\\
w \in \L(\SGf(\xi,\underline{\R}))  \implies & w \models \Gf{\bowtie p}_\ext \xi & \implies w \in \L(\SGf(\xi,\overline{\R})) \label{eq:sl6}
\end{eqnarray}
	Moreover, the result holds also for $\underline{\R} \cap \sf(\xi) \subseteq \R(w) \cap \sf(\xi) \subseteq \overline{\R} \cap \sf(\xi)$. 
\end{reflemma}
\begin{proof}
Due to Lemma~\ref{lem:slave-promises}, it suffices to prove for the given $\xi$ and $w$ and for any $\R$ that
\begin{align}
	\sat({\R}) \text{ is infinite} & \iff w \in \L(\SGF(\xi,{\R})) \label{eq:sl11}\\
	\Nset \setminus \sat(\R) \text{ is finite}  & \iff w \in \L(\SFG(\xi, \R)) \label{eq:sl12}\\
    \lrExt(\big(\idf_{\sat(\R)}(i)\big)_{i=0}^\infty) \bowtie p & \iff w \in \L(\SGf(\xi, \R)) \label{eq:sl13}
\end{align}

For (\ref{eq:sl11}), we must prove that there are infinitely many positions from which the run ends in an accepting sink if{}f there are infinitely many positions with a token in an accepting sink. 
To this end, observe that to each position $j$ with a token in an \emph{accepting} sink $q$ (i.e., $\R\proves q$.) we can assign a set $\mathit{EndIn}(j,q)$ of positions $i$ such that $\SGF(\xi)(\infix w ij)=q$.
On the one hand, each $i$ is exactly in one $\mathit{EndIn}(j,q)$ since the slave transition systems are acyclic and each path inevitably ends in a sink.
On the other hand, each $\mathit{EndIn}(j,q)$ is finite, again due to the acyclicity. 
Consequently, $\sat({\R})$ is infinite if{}f $\sum_{j,q}\mathit{EndIn}(j,q)$ is infinite if{}f the number of non-empty $\mathit{EndIn}(j,q)$ is infinite if{}f $\SGF$  accepts.

For (\ref{eq:sl12}), the argument is analogous, but we have to consider $\mathit{EndIn}(j,q)$ for \emph{rejecting} sinks $q$, i.e., $\R\not\proves q$.
Then $\sat({\R})$ is co-finite if{}f $\sum_{j,q}\mathit{EndIn}(j,q)$ is finite if{}f the number of non-empty $\mathit{EndIn}(j,q)$ is finite if{}f $\SFG$ accepts.

For (\ref{eq:sl13}), observe that in $\SGf$ the precise number of tokens is preserved in each state at every point of time.
Therefore, each successful run corresponds exactly to one $1$ in the total reward.
In order to prove that both sequence have the same $\liminf/\limsup$, we need to prove that the length of each run (difference between the element's positions in the two sequences) is bounded.
This follows by acyclicity of the automaton.
\qed	
\end{proof}

\begin{reflemma}{lem:slaves-product}
For $w$ and $\R \subseteq \rec$, we have
\begin{description}
\item[(soundness)]  whenever
        $w \in \L(\P({\R}))$ then $\R\subseteq\R(w)$ and hence
        \[w \models \bigwedge_{\F\xi \in \R} \G\F \xi \land \bigwedge_{\G\xi \in \R} \F\G \xi \land \bigwedge_{\Gf{\bowtie p}_\ext \xi \in \R} \Gf{\bowtie p}_\ext \xi\]
\item[(completeness)]    $w \in \L(\P({\R(w)}))$.    
\end{description}
\end{reflemma}
\begin{proof}
As to soundness, let $w \in \L(\S({\R}))$. Consider the dag on $\R$ given by an edge $(\chi,\chi')$ if $\chi'\in\sf(\chi)\setminus\{\chi\}$.
We prove the right-hand side of the implication for each formula $\xi\in\R$ by induction on the distance $d$ to the leaf in the dag.

Let $d=0$ and consider $\chi=\F\xi$; the other cases are analogous. 
Then $\xi$ does not contain any subformula from $\R$.
Therefore, not only $w\in\L(\SGF(\xi,\R)$, but also $w\in\L(\SGF(\xi,\emptyset)$.
Since $\emptyset\subseteq\R(w)$, Lemma~\ref{lem:slaves-acc} (part ``Moreover'') yields $w\models\G\F\xi$.

Let $d>0$ and $\chi=\F\xi$; the other cases are again analogous. 
We have not only $w\in\L(\SGF(\xi,\R)$, but also $w\in\L(\SGF(\xi,\R\cap\sf(\xi))$.
By induction hypothesis, $w\models\R\cap\sf(\xi)$.
Therefore, $\R\cap\sf(\xi)\subseteq\R(w)$ and thus Lemma~\ref{lem:slaves-acc} yields $w\models\G\F\xi$.

As to completeness, we prove that $w\in\L(\xi,\SGF(\R(w)))$ for $\F\xi\in\R(w)$; the proof for other types of automata is analogous. 
Since $\F\xi\in\R(w)$ we have $w\models\G\F\xi$.
By Lemma~\ref{lem:slaves-acc} we have $w \in \L(\SGF(\xi,\R(w)))$.
\qed
\end{proof}

We call the left-hand-side of $\proves$ of the acceptance condition ``extended assumptions'' since it is a conjunction of assumptions $\R$ extended by $\Psi_\xi[(\rec\setminus\R)/\false]$ for each $\G\xi\in\R$. We prove the extended assumptions hold at almost all positions:

\begin{alemma}\label{lem:tokens}
For every word $w$ accepted with respect to $\R$, and for every formula $\G\xi \in \R$, and for all $t\geq T(w)$ and for all $\tau\geq t$ such that $\psi:=\SGF(\xi)(\infix w t \tau)$ is defined, we have that 
$\suffix w\tau \models \psi[\rec\setminus\R/\false]$.
\end{alemma}
\begin{proof}
Any tokens born at time $t\geq T(w)$ will end up at some time $\delta \geq t$ in an accepting sink $s_\delta$.
Since $\suffix w {\delta}\models \R(w)$ and $\R(w)\proves s_\delta$, we have $\suffix w {\delta}\models s_\delta$.
Since also $\suffix w\tau\models \R(w)$ for all $\tau\in[t,\delta]$, we obtain also  $\suffix w \tau\models \psi$ by similar argumentation as in Lemma~\ref{lem:threshold}.

Moreover, since $\R(w)\proves s_t$, by propositional calculus $\R(w)\proves s_\delta[\rec\setminus\R(w)/\false]$ and $\suffix w\delta\models s_\delta[\rec\setminus\R(w)/\false]$, and similarly we obtain $\suffix w {\tau}\models \psi[\rec\setminus\R(w)/\false]$.
\qed
\end{proof}

\begin{refproposition}{prop:A-sound}
	If $w \in \L(\A)$, then $w \models \varphi$.
\end{refproposition}
\begin{proof}
If $w \in \L(\A)$ and it accepts by a disjunct in its acceptance conditions related to assumptions $\R \subseteq \rec$, then for almost all positions $t$ when visiting a state $(\psi,(\Psi_\xi)_{\xi\in \rec})$ we have 
\[\R\cup \bigcup_{\G\xi\in\R} \Psi_\xi[\rec\setminus\R/\false] \proves \psi\] 
and, moreover by Lemma~\ref{lem:slaves-product} and \ref{lem:tokens}, we also have \[\suffix w t\models\R\cup \bigcup_{\G\xi\in\R} \Psi_\xi[\rec\setminus\R/\false] \]
yielding together by (\ref{eq:transitivity})
\[\suffix w t\models \psi\]
which by Lemma~\ref{lem:master-local} gives $w\models \varphi$.
\qed
\end{proof}

\begin{refproposition}{prop:A-complete}
	If $w \models \varphi$, then $w \in \L(\A)$.
\end{refproposition}
\begin{proof}
Let $w$ be a word. Then $\acc(\R(w))$ is satisfied by Lemma~\ref{lem:slaves-product}.
We show that $\acc_{\T}(\R(w))$ is satisfied, too.
In other words, we prove that for almost all all positions $t$ when visiting a state $(\psi,(\Psi_\xi)_{\xi\in \rec})$ we have
\[\R(w)\cup \bigcup_{\G\xi\in\R(w)} \Psi_\xi[\rec\setminus\R(w)/\false] \proves \psi\]

Since both $\psi$ and each element of each $\Psi_\xi$ are actually Boolean functions, we choose formulae that are convenient representations thereof.
Namely, we consider the formula generated exactly from $\varphi$ or $\xi$ using the transition functions $\delta^\T$ or $\delta^\S$, respectively.
Therefore, each occurrence of $\G\xi\in\sf(\varphi)$ corresponds after reading a finite word $v$ to some occurrence of $\psi'\in\sf(\psi)$ where $\psi'=\G\xi\wedge \bigwedge_i \xi_i$ and $\xi_i=\delta^\T(\xi,v_i)$ for some infix $v_i$ of $w$; we call such a formula $\psi'$ \emph{derived $\G$-subformula}.
Similarly, reading $v$ transforms $\F\xi$ into a \emph{derived $\F$-subformula} $\F\xi\vee\bigvee_i \xi_i$.
Finally, similarly for $\Gf{\bowtie}_\ext$- and $\U$-formulae.
Note that every derived $\Gf{\bowtie p}_\ext$-formula is always of the form $\Gf{\bowtie p}_\ext \xi \wedge\bigwedge \true$.

We consider positions large enough so that 
\begin{itemize}
\item they are greater than $T(w)+|Q|$ (here $|Q|$ ensures that tokens born before $T(w)$ do not exist any more), and
\item all the satisfied $\U$-formulae have their second argument already satisfied, and
\item $\psi$ is a Boolean combination over derived formulae since all outer literals and $\X$-operators have been already removed through repetitive application of $[\cdot]$.
\end{itemize}

We prove that each derived formula $\psi'$ (in $\psi$) that currently holds is also provable from the extended  assumptions.
Since $\psi$ holds, this implies that also the whole $\psi$ is provable from the extended  assumptions.
We proceed by structural induction.

First, let $\psi'$ be a derived $\G$-subformula $\G\xi\wedge \bigwedge_i \xi_i$.
Since $\psi'$ holds, $\G\xi$ holds, we have $\G\xi\in\R(w)$ and thus $\R(w)\proves\G\xi$.
Further, each $\xi_i$ corresponds to a formula $\psi_i$ either in $\Psi_\xi$ or a sink, which is accepting since $\xi_i$ holds, as follows:
This correspondence mapping is very similar to identity, except for
\begin{itemize}
\item each derived $\F$-formula $\F\chi\vee\bigvee_i \chi_i$ is mapped to $\F\chi$ since $\S(\xi)$ does not unfold $\F$, and
\item each derived $\G$-formula $\G\chi \wedge \bigwedge_i \chi_i$ is mapped to $\G\chi$ since $\S(\xi)$ does not unfold $\G$, moreover, each $\chi_i$ again corresponds in the same way to a formula in $\Psi_\chi$ or an accepting sink of $\SGF(\R(w),\chi)$ by the induction hypothesis.
\end{itemize}
If we could replace each derived formula in $\xi_i$ by its simple image in the correspondence mapping, we would have $\psi_i\proves\xi_i$ (and since $\psi_i$ is provable from assumptions - either a token or an accepting sink - we could conclude).
Therefore it remains to prove all the derived formulae:
\begin{itemize}
\item $\G\chi \wedge \bigwedge_i \chi_i$ that holds can be proved by induction hypothesis,
\item $\G\chi \wedge \bigwedge_i \chi_i$ that does not hold is proved from $\G\chi[\rec\setminus\R(w)/\false]=\false$
\item $\F\chi\vee\bigwedge_i \chi_i$ that holds is proved from $\F\chi\in\R(w)$
\item $\F\chi\vee\bigwedge_i \chi_i$ that does not hold is proved from $\F\chi[\rec\setminus\R(w)/\false]=\false$
\end{itemize}

Second, $\psi'=\Gf{\bowtie p}_\ext\xi\wedge\bigwedge\true$ is proved directly from $\R(w)$.

Third, let (i) $\psi'$ be a derived $\F$-subformula $\F\xi\vee \bigvee_i \xi_i$ such that $\F\xi$ holds. 
Then $\F\xi\in\R(w)$ and thus $\R(w)\proves\F\xi$.

Finally, let $\psi'$ be a derived $\F$-subformula $\F\xi\vee \bigvee_i \xi_i$ such that $\F\xi$ does not hold (i.e., some of the $\xi_i$'s hold), or a derived $\U$-subformula, where thus one of the disjuncts not containing this until holds (since all satisfied untils have their second argument already satisfied). Then we conclude by the induction hypothesis.
\qed
\end{proof}

\section{Probability space of Markov chain}
\label{app:mc}
For a Markov chain $N=(L,P,\hat \ell)$ we define the probability space $(\mathit{Run}, \mathcal{F}, \mathbb{P})$ where
\begin{itemize}
 \item $\mathit{Run}$ contains all runs initiated in $\hat \ell$, i.e.
  all infinite sequences $\ell_0\ell_1\ldots$ satisfying $\ell_0=\hat \ell$ and $P(\ell_i,\ell_{i+1}) > 0$ for all $i\ge 0$.
 \item $\cal F$ is the $\sigma$-field generated by basic cylinders $\mathit{Cyl}(h) := \{\omega \mid \omega \text{ starts with } h\}$ for all $h$ which are a prefix of an element in $\mathit{Run}$.
 \item $\mathbb{P}$ is the unique probability function such that for $h=\ell_0\ell_1\ldots \ell_n$ we have $\mathbb{P}(\mathit{Cyl}(h)) = \prod_{i=0}^{i=n-1} P(\ell_i,\ell_{i+1})$
\end{itemize}

When we say ``almost surely'' or ``almost all runs'', it refers to an event happening with probability 1 according to the relevant measure (which is usually clear from the context).

\section{Proof of Proposition~\ref{thm:mp-main}}
In the rest of this section we prove Proposition~\ref{thm:mp-main}. 
To simplify the notation, for an action $a$ and reward structure $r$ we will use $\lrExt(a)$ and $\lrExt(r)$ for random variables that on a run $\omega = s_0a_0s_1a_1\ldots$ 
return $\lrExt(\idf_{a_0=a}\idf_{a_1=a}\ldots)$ and $\lrExt(r_{s_0}r_{s_1}\ldots)$, respectively.

The direction $\Rightarrow$ can be proved as follows.
For any fixed $i$, by \cite[Corollary 12]{lics15} there is a strategy $\sigma_i$ such that $\lrLim{a} = x_{i,a}$ almost surely, and hence $\sigma_i$
almost surely yields reward $\sum_{a\in A} \singlereward(a)\cdot x_{i,a}$ w.r.t.\ any reward function $\singlereward$. Moreover, $\sigma_i$ visits
every state of $\mdp$ infinitely often almost surely.

We now construct $\sigma$ inductively as follows. The strategy will keep the current ``mode'', which is a number from $1$ to $n$, and an unbounded ``timer'' ranging over natural numbers. Suppose we have
defined $\sigma$ for history $h$, but not for any other history starting with $h$.
Suppose that in the history before $h$ the 
strategy $\sigma$ was in mode $\ell$. Then in $h$ the mode is incremented by $1$ (modulo $n$), yielding the mode $\ell'$, and the strategy $\sigma$ starts playing as
$\sigma_{\ell'}$. It does so for $2^{|h|}$ steps, yielding a history $h'$. Afterwards, we apply the inductive definition again with $h'$ in place of $h$. 

\begin{lemma}
The strategy $\sigma$ satisfies the requirements of Proposition~\ref{thm:mp-main}.
\end{lemma}
\begin{proof}[Sketch]
Firstly, the generalised B\"uchi condition is almost surely satisfied because it is satisfied under any $\sigma_i$, and $\sigma$ will eventually mimic $\sigma_i$
for an arbitrary long length.

Let us continue with the claim for the mean payoffs with supremum limits. Fix $1\le i\le n$. We will show that for every $\varepsilon>0$, almost every run $\omega$
has a prefix $s_0a_0s_1a_1\ldots s_\ell$
with 
\[
\frac{1}{\ell}\sum_{j=0}^{\ell-1} \singlerewardalt_i(a_j) \ge u_i - \varepsilon
\]

By properties of $\sigma_i$ for any $s\in S$ there is a number $k_{s,\varepsilon}$ and a set of runs 
$R$ with $\Pr{\sigma_i}{s}{T} \ge 1/2$ such that
for every $s'_0a'_0s'_1a'_1\ldots \in T$ and every $k' \ge k_{s,\varepsilon}$ we have 
\[
 \sum_{\ell=0}^{k'} \frac{1}{k'} \singlerewardalt_i(a'_\ell) \ge u_i - \varepsilon/2 
\]
Let $\alpha$ be the smallest assigned reward, there must be a number $J_\varepsilon$ such that
\[
  J_\varepsilon \cdot (u_i - \alpha) < 2^{J_\varepsilon} \cdot \varepsilon/2
\]
Intuitively, $J_\varepsilon$ is chosen so that no matter what the history is in the first $J_\varepsilon$ steps,
if the remainder has length at least $2^{J_\varepsilon}$ steps and gives partial average at least $u_i - \varepsilon/2$,
we know that the whole history gives a partial average at least $u_i - \varepsilon$.

Now almost every run $\omega$ has infinitely many prefixes $h_0,h_1\ldots$ such in the prefix $h_i$, the strategy $\sigma$ starts
mimicking $\sigma_i$ for $2^{|h_i|}$ steps. Now consider those prefixes $h_i$ which have length greater than $\max_s k_{\varepsilon,s}$
and $J_\varepsilon$. This ensures that starting with any such prefix $h_i$, with probability at least $1/2$ the history $h'=s'_0a'_0s'_1a'_1\ldots s'_\ell$ in which we end after
taking $2^{|h_i|}$ steps will satisfy
\[
 \frac{1}{\ell}\sum_{j=0}^{\ell-1} \singlerewardalt_i(a'_j) \ge u_i - \varepsilon
\]
Using Borel-Cantelli lemma~\cite{Royden88} this implies that almost every run has the required property.

The proof for mean payoff with inferior limits is analogous, although handling limit inferior is more subtle as it requires us to show
that from some point on the partial average never decreases below a given bound. To give a formal proof, we can reuse
the construction from \cite[Proof of Claim 10]{lics15} applied to strategies $(\xi_k)_{1\le k\le \infty}$ where each $\xi_k$ for $k$ of the form $\ell\cdot j + i$
is defined to be the strategy $\sigma_i$. Note that our choice of ``lengths'' of each mode
satisfy Equations (3) and (4) from \cite[Proof of Claim 10]{lics15}. Also note that while \cite[Proof of Claim 10]{lics15} requires the frequencies
of the actions to converge, in our proof we are only concerned about limits inferior of long run rewards, and so the requirements on $\xi_i$ are not
convergence of limits,
but only that limits inferior converge to the required bound. This requirement is clearly satisfied.
\qed
\end{proof}

Let us now proceed with the direction $\Leftarrow$ of the proof or Proposition~\ref{thm:mp-main}.
Because no $x_{i,a}$ and $x_{i',a'}$ with $i\neq i'$ occur in the same equation, we can
fix $1\le i \le m$ and to finish the proof it suffices to give a solution to $x_{i,a}$ for all $a$.

Similarly to~\cite{lics15,BBC+14} where only lim inf was considered, the main idea of the proof is to obtain suitable ``frequencies''
of actions and use these as the solution. Nevertheless, the formal approach of~\cite{lics15,BBC+14} itself cannot be easily adapted
(main issue being the use of Fatou's lemma, which for the purpose of limit superior does not allow to establish the required
inequality). Instead, we use a straightforward adaptation of an approach used in~\cite{BCFK13}. The statement we require is captured in the following lemma.

\begin{lemma}\label{lemma:subsequence}
For every run $\omega=s_0a_0s_1a_1\ldots$ %
there is a sequence of numbers $T_1[\omega],T_2[\omega],\ldots$ such that the number
\[
f_{\omega}(a) := \lim_{\ell\rightarrow \infty} \frac{1}{T_\ell[\omega]} \sum_{j=1}^{T_\ell[\omega]} \idf_{a_j=a}
\]
is defined and non-negative for all $a\in A$, and satisfies
\[
\begin{array}{rl}
\sum_{a\in A} f_{\omega}(a)\cdot \singlerewardalt_i(a) &= \lrSup(\singlerewardalt_i)(\omega)\\
\sum_{a\in A} f_{\omega}(a)\cdot \singlereward_j(a) &\ge \lrInf(\singlereward_j)(\omega)\quad{\text{for $1\le j \le n$}}\\
\sum_{a\in A} f_{\omega}(a) &= 1
\end{array}
\]
Moreover, for almost all runs $\omega$ we have
\[
 \sum_{a\in A} f_{\omega}(a) \cdot \delta(a)(s) = \sum_{a\in \act{s}} f_{\omega}(a)
\]
\end{lemma}
\begin{proof}
Fix $\omega=s_0a_0s_1a_1\ldots$
We first define a sequence $T'_1[\omega],T'_2[\omega],\ldots$ to be any sequence satisfying 
\[
\lim_{\ell\rightarrow \infty} \frac{1}{T'_\ell[\omega]} \sum_{j=1}^{T'_\ell[\omega]} \singlerewardalt_i(a_j)\quad = \lrSup(\singlerewardalt_i)(\omega)
\]
Existence of such a sequence follows from the fact that every sequence of real numbers has a subsequence which converges to the lim sup of the original sequence.

Further, we define subsequences $\hat T^k_1[\omega],\hat T^k_2[\omega],\ldots$ for $1\le k \le n$
where for all $k$ the sequence $\hat T^k_1[\omega],\hat T^k_2[\omega],\ldots$ satisfies
\[
\lim_{\ell\rightarrow \infty} \frac{1}{\hat T^k_i[\omega]} \sum_{j=1}^{\hat T^k_i[\omega]} \singlerewardalt_i(a_j)\quad = \lrSup(\singlerewardalt_i)(\omega)
\]
and
\[
\lim_{\ell\rightarrow \infty} \frac{1}{\hat T^k_i[\omega]} \sum_{j=1}^{\hat T^k_i[\omega]} \singlereward_{k'}(a_j)\quad \ge \lrInf(\singlereward_{k'})(\omega)
\]
for all $k' \le k$. We define these subsequences inductively. We start with $\hat T^0[\omega],\hat T^1[\omega]\ldots = T'_1[\omega],T'_2[\omega]\ldots$.
Now assuming that $\hat T^{k-1},\hat T^{k-1}\ldots = 0,1\ldots$ has been defined, we
take $\hat T^k_1[\omega],\hat T^k_2[\omega],\ldots$ such that
\[
\lim_{\ell\rightarrow \infty} \frac{1}{\hat T^k_i[\omega]} \sum_{j=1}^{\hat T^k_i[\omega]} \singlereward_k(a_j)
\]
exists. The existence of such a sequence follows from the fact that every sequence of real numbers has a converging subsequence. The required properties then follow easily from properties of limits.

Now assuming an order on actions, $\bar a_1,\ldots,\bar a_{|A|}$ in $A$, we define $T^{k}_1[\omega],T^{k}_2[\omega],\ldots$ for $0\leq k\leq |A|$ so that $T^{0}_1[\omega],T^{0}_2[\omega],\ldots$ is the sequence $\hat T^n_1[\omega],\hat T^n_2[\omega],\ldots$, and every $T^{k}_1[\omega],T^{k}_2[\omega],\ldots$ is a subsequence of $T^{k-1}_1[\omega],T^{k-1}_2[\omega],\ldots$ 
such that the following limit exists
\[
f_\omega(\bar a_k) := \lim_{\ell\rightarrow \infty} \frac{1}{T^{k}_i[\omega]} \sum_{j=1}^{T^{k}_i[\omega]} \idf_{a_j=\bar a_k}
\]
The required properties follow as before. We take $T^{|A|}_1[\omega],T^{|A|}_2[\omega],\ldots$ to be the desired sequence $T_1[\omega],T_2[\omega],\ldots$.

Now we need to show that satisfies the required properties. Indeed
\begin{align*}
\sum_{a\in A} f_{\omega}(a)\cdot \singlerewardalt_i(a) &=  \sum_{a\in A} \lim_{\ell\rightarrow \infty} \frac{1}{T^{k}_\ell[\omega]} \sum_{j=1}^{T^{k}_\ell[\omega]} \idf_{a_j = a} \cdot \singlerewardalt_i(a)
 \tag{def. of $f_{\omega}(a)$}\\
 &= \lim_{\ell\rightarrow \infty} \frac{1}{T^{k}_\ell[\omega]} \sum_{j=1}^{T^{k}_\ell[\omega]} \singlerewardalt_i(a)
 \tag{property of $\idf$ and the sum}\\
 &= \lrSup(\singlerewardalt_i)(\omega)\tag{def. of subsequence $T^k_\ell[\omega]$}
\end{align*}
and analogously, for any $1\le i'\le n$:
\begin{align*}
\sum_{a\in A} f_{\omega}(a)\cdot \singlereward_{i'}(a) &=  \sum_{a\in A} \lim_{\ell\rightarrow \infty} \frac{1}{T^{k}_\ell[\omega]} \sum_{j=1}^{T^{k}_\ell[\omega]} \idf_{a_j = a} \cdot \singlereward_{i'}(a)\\
 &= \lim_{\ell\rightarrow \infty} \frac{1}{T^{k}_\ell[\omega]} \sum_{j=1}^{T^{k}_\ell[\omega]} \singlereward_{i'}(a)\\
 &\ge \lrInf(\singlereward_{i'})(\omega)
\end{align*}
Also
\[
\sum_{a\in A} f_{\omega}(a) \quad= \sum_{a\in A} \lim_{\ell\rightarrow \infty} \frac{1}{T^{k}_\ell[\omega]} \sum_{j=1}^{T^{k}_\ell[\omega]} \idf_{a_j = a}
 \quad= \lim_{\ell\rightarrow \infty} \frac{1}{T^{k}_\ell[\omega]} \sum_{j=1}^{T^{k}_\ell[\omega]} 1 \quad= 1\\
\]

To prove the last property in the lemma, we invoke the law of large numbers (SLLN) \cite{KSK76}. Given a run $\omega$, an action $a$, a state $s$ and $k\geq 1$, define
\[
N^{a,s}_k(\omega)=\begin{cases}
  1 & \text{ $a$ is executed at least $k$ times}\\&\text{ and $s$ is visited just after the $k$-th execution of $a$; }\\
  0 & \text{ otherwise.}
\end{cases}
\]
By SLLN and by the fact that in every step the distribution on the next states depends just on the chosen action, for almost all runs $\omega$ the following limit is defined and the equality holds whenever $f_\omega(a) > 0$:
\[
\lim_{j\rightarrow \infty} \frac{\sum_{k=1}^j N^{a,s}_k(\omega)}{j} = \delta(a)(s)
\]
We obtain, for almost every $\omega=s_0a_0s_1a_1\ldots$
\begin{eqnarray*}
\lefteqn{\sum_{a\in A} f_{\omega}(a)\cdot \delta(a)(s)}\\& = &  \sum_{a\in A} \lim_{\ell\rightarrow \infty} \frac{1}{T_\ell[\omega]} \sum_{j=1}^{T_\ell[\omega]} \idf_{a_j=a}\cdot   
  \lim_{\ell\rightarrow \infty} \frac{1}{\ell}\sum_{k=1}^{\ell} N^{a,s}_k(\omega) \\
& = & \sum_{a\in A} \lim_{\ell\rightarrow \infty} \frac{1}{T_\ell[\omega]} \sum_{j=1}^{T_\ell[\omega]}  \idf_{a_j=a}\cdot \lim_{\ell\rightarrow \infty} \frac{1}{\sum_{j=1}^{T_\ell[\omega]} \idf_{a_j=a}}\sum_{k=1}^{\sum_{j=1}^{T_\ell[\omega]} \idf_{a_j=a}} N^{a,s}_k(\omega) \\
& = & \sum_{a\in A} \lim_{\ell\rightarrow \infty} \frac{1}{T_\ell[\omega]}\sum_{k=1}^{\sum_{j=1}^{T_\ell[\omega]} \idf_{a_j=a}} N^{a,s}_k(\omega) \\
& = & \lim_{\ell\rightarrow \infty} \frac{1}{T_\ell[\omega]}\sum_{a\in A} \sum_{k=1}^{\sum_{j=1}^{T_\ell[\omega]} \idf_{a_j=a}} N^{a,s}_k(\omega) \\
& = & \lim_{\ell\rightarrow \infty} \frac{1}{T_\ell[\omega]}\sum_{j=1}^{T_\ell[\omega]} \idf_{s_j=s} \\
& = & \lim_{\ell\rightarrow \infty} \frac{1}{T_\ell[\omega]}\sum_{j=1}^{T_\ell[\omega]} \sum_{a\in \mathit{Act}(s)} \idf_{a_j=a} \\
& = & \sum_{a\in \mathit{Act}(s)}\lim_{\ell\rightarrow \infty} \frac{1}{T_\ell[\omega]}\sum_{j=1}^{T_\ell[\omega]} \idf_{a_j=a} \\
& = & \sum_{a\in \mathit{Act}(s)} f_{\omega}(a)
\end{eqnarray*}
\qed
\end{proof}

We apply Lemma~\ref{lemma:subsequence} to obtain values $f_\omega$ for every $\omega$.
Now it suffices to consider any $\omega$ for which $f_\omega$ satisfies the last condition of the lemma and which
also satisfies $\lrInf(\singlereward_j[\omega]) \ge u_j$ for all $1\le j\le n$ and $\lrSup(\singlerewardalt_i[\omega]) \ge v_i$; by the assumptions on $\sigma$ and $R$ such a run must exist.
This immediately gives us that all the equations from Figure~\ref{system-L} are satisfied.

\fi

\end{document}